\documentclass[12pt]{article}
\usepackage{amsmath}
\usepackage{graphicx,psfrag,epsf,amsfonts,mathtools,amsthm}
\usepackage{enumerate}
\usepackage[numbers]{natbib}

\usepackage{algorithm, algpseudocode} 
\usepackage[toc,page]{appendix}

\usepackage[locale=FR]{siunitx}
\usepackage{booktabs}
\usepackage[table,xcdraw]{xcolor}

\usepackage{caption}
\usepackage{subcaption}

\usepackage{enumitem} 

\usepackage{tikz}
\usetikzlibrary{automata, positioning, arrows, shapes}
\newcommand{\blind}{0}

\newtheorem{theorem}{Theorem}

\usepackage{mathtools} 

\begin{document}

\def\spacingset#1{\renewcommand{\baselinestretch}%
{#1}\small\normalsize} \spacingset{1}


\if0\blind
{
  \title{\bf A Case-Study of Sample-Based Bayesian Forecasting Algorithms}
  \author{Taylor R. Brown
  \hspace{.2cm}\\
    Department of Statistics, University of Virginia}
  \maketitle
} \fi

\if1\blind
{
  \bigskip
  \bigskip
  \bigskip
  \begin{center}
    {\LARGE\bf Title}
\end{center}
  \medskip
} \fi

\bigskip
\begin{abstract}
For a Bayesian, real-time forecasting with the posterior predictive distribution can be challenging for a variety of time series models. First, estimating the parameters of a time series model can be difficult with sample-based approaches when the model's likelihood is intractable and/or when the data set being used is large. Second, once samples from a parameter posterior are obtained on a fixed window of data, it is not clear how they will be used to generate forecasts, nor is it clear how, and in what sense, they will be ``updated" as interest shifts to newer posteriors as new data arrive. This paper provides a comparison of the sample-based forecasting algorithms that are available for Bayesians interested in real-time forecasting with nonlinear/non-Gaussian state space models. An applied analysis of financial returns is provided using a well-established stochastic volatility model. The principal aim of this paper is to provide guidance on how to select one of these algorithms, and to describe a variety of benefits and pitfalls associated with each approach.
\end{abstract}

\noindent%

\spacingset{1.45}
\section{Introduction}
\label{sec:intro}

In this paper, we consider the problem of forecasting from a Bayesian perspective. We concern ourselves with approximating, at each point in time, the \emph{posterior predictive distribution}, using a variety of online sample-based algorithms for nonlinear and/or non-Gaussian state space models. We examine algorithms that feature some combination of particle filters and Markov chain Monte Carlo samplers. 




The first (and simplest) approach used in this paper consists of using a single particle filter with its static parameters set to posterior point estimates. If we treat these parameters as ``known," particle filters will readily provide forecasts. With enough data, if the parameter posterior becomes informative enough, ignoring the uncertainty of point estimates might not be an unreasonable assumption, practically speaking. 

A second approach considered in this paper is the \emph{particle swarm filter} \cite{pswarm} (not to be confused with particle swarm optimization). This idea continues in the spirit of the first approach, but takes into account parameter uncertainty. Again, a separate algorithm obtains samples from a parameter posterior distribution after enough observed data has been accumulated. Then these samples are used to instantiate many particle filters, whose outputs are averaged together. Parameter posterior samples might be refreshed by successive runs of an MCMC sampler. 

A third approach avoids the requirement of using a secondary algorithm that samples from the parameter posterior. Imagine running a filtering algorithm for a model with a state vector augmented with the unknown parameters \cite{selforganizing}, \cite{liuandwest}. After iterating through enough data, the marginal filtering distribution will ideally resemble the true parameter posterior, and the forecasts will hopefully resemble the true posterior predictive distribution. 

The strengths and weaknesses of these approaches will be investigated in the case study in section \ref{sec:case_study}. Prior to that, section \ref{sec:models} will describe a high level view of the models, and then all of our algorithms will be described in section \ref{sec:algos}. 

Finally, it should be mentioned that there are a number of other extremely popular sample-based approaches for obtaining forecasts under parameter uncertainty \cite{mcmcsuffstatspfs}, \cite{resamplemove}, \cite{smcsquared} and \cite{storvik}. All of these either assume that the state space model is of a certain form, or they are not online. For our particular discussion, these are disqualifying criteria, and so they will not be considered any further in this paper.

\section{Models}
\label{sec:models}
\subsection{State Space Models} \label{subsec:ssms}

Let $T \in \mathbb{N}$ and for $t=1,\ldots,T$, define $y_t$ to be an observable random variable and $x_t$ to be an unobserved/latent random variable. In the case of our model, all of these will be $\mathbb{R}$-valued. Denote the two sequences as $y_{1:T}$ and $x_{1:T}$. 

A state space model is a partially-observed Markov chain with additional assumptions of conditional independence. The observations are usually assumed to be independent after conditioning on all of the states, and these states are usually assumed to be marginally Markovian themselves. For our particular model, though, the state transitions depend on lagged observations. Given a vector of real-valued parameters $\theta$, we will call the first time's state distribution $\mu(x_1 \mid \theta)$, each state transition density $f(x_t \mid x_{t-1}, y_{t-1} \theta)$, and the observation densities $g(y_{t} \mid x_{t})$. 

\begin{figure}
\centering
\tikzset{rectangle state/.style={draw,rectangle}} 
\begin{tikzpicture}[scale=0.2]
  \node[state, minimum size=1.5cm] (y1) {$y_{1}$};
  \node[rectangle state,minimum size=1.5cm] (x1) [below= of y1] {$x_1$};
  \node[state, minimum size=1.5cm] (y2) [right= of y1] {$y_{2}$};
  \node[rectangle state, minimum size=1.5cm] (x2) [right= of x1] {$x_{2}$};
  \node[state, draw=none, minimum size=1.5cm] (ykdots) [right= of y2] {$\cdots$};
  \node[state, draw=none, minimum size=1.5cm] (xkdots) [right= of x2] {$\cdots$};
  \node[state, minimum size=1.5cm] (ynm1) [right= of ykdots] {$y_{n-1}$};
  \node[rectangle state, minimum size=1.5cm] (xnm1) [below= of ynm1] {$x_{n-1}$};
  \node[state, minimum size=1.5cm] (yn) [right= of ynm1] {$y_{n}$};
  \node[rectangle state, minimum size=1.5cm] (xn) [below= of yn] {$x_{n}$};
  \draw[->] 
    (x1) edge[left] (y1)    
    (x2) edge[left] (y2)
    (xnm1) edge[left] (ynm1)
    (xn) edge[left] (yn)
    (x1) edge[left] (x2)
    (y1) edge[left] (x2)
    (x2) edge[left] (xkdots)
    (xkdots) edge[left] (xnm1)
    (ykdots) edge[left] (xnm1)
    (xnm1) edge[left] (xn)
    (ynm1) edge[left] (xn);
\end{tikzpicture}
\caption{Diagram of a SSM}
\label{fig:ssm_diagram}
\end{figure}
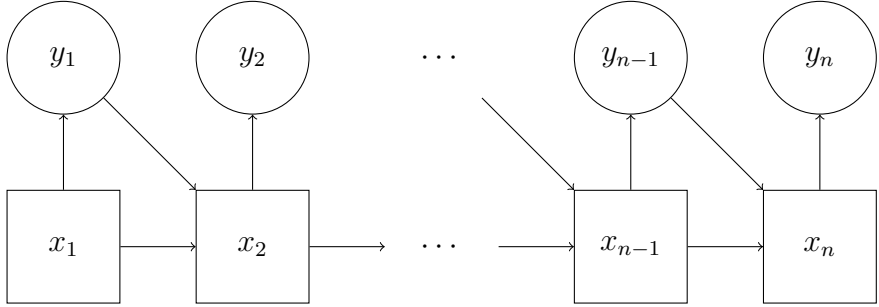

When the issue of parameter uncertainty remains unaddressed, and the primary goal is state inference in real-time, \emph{filtering distributions} are the primary object of interest. A filtering distribution is the distribution of an unknown state, at a given point in time, conditional on all of the information obtained up to that point in time: $p(x_t \mid y_{1:t}, \theta)$. Filtering distributions can be expressed recursively in time using Bayes' rule. 

Particle filters, explained below in section \ref{subsec:particle_filtering}, provide recursive sample-based approximations to expectations of these distributions. They also provide an approximation to the observed data likelihood $p(y_{1:t} \mid \theta)$, a quantity which, in our case, is an intractable integral of the complete-data likelihood. 

\subsection{Parameter Uncertainty} \label{subsec:param_uncertainty}

The parameters of state space models are usually unknown in practice. In these cases, parameter inference is necessary. In Bayesian statistics, a primary goal of inference is the \emph{posterior distribution} $p(\theta \mid y_{1:t})$. Because we are interested in forecasting, our situation is complicated by the fact that the amount of data we possess varies as time $t$ increases. At each moment in time, there is a new targetable posterior. Using Bayes' formula (and assuming without loss of generality that the parameters are continuous) we can write it as 

\begin{eqnarray}
p(\theta \mid y_{1:t} ) = \frac{ p(y_{1:t} \mid \theta) p(\theta) }{ \int p(y_{1:t} \mid \vartheta) p(\vartheta) d\vartheta}. \label{eqn:posterior}
\end{eqnarray}

The primary purpose of this paper is forecasting, though, so we will be primarily interested in the \emph{posterior predictive distributions} (ppds). These are quantities that depend on a sequence of posterior distributions. At each moment in time, an \emph{up-to-date} posterior is used to average over the parameter space:
\begin{eqnarray}
p(y_{t+1} \mid y_{1:t}) = \int p(y_{t+1} \mid \theta, y_{1:t}) p(\theta \mid y_{1:t} ) d \theta. \label{eqn:ppd}
\end{eqnarray}
This paper is concerned with nonlinear and/or non-Gaussian state space models, so both factors in the integral of equation \ref{eqn:ppd} are intractable, in addition to being time dependent. 

Note that, for any fixed segment of data $y_{1:t}$, a plethora of Markov chain Monte Carlo algorithms are available \cite{mcmcforssms}, as well as other approaches that fall under the category of variational inference \cite{Blei2017}. The former are all sample-based, and do not require evaluations of the normalized posterior. We will focus on those in this paper. A specific algorithm that can be used to estimate the posterior distribution of the model we consider will be described in section \ref{subsec:mcmc}.


\subsection{A Stochastic Volatility Model}
\label{sec:svol_mod}

This paper compares forecasting performance of different algorithms using the model of \cite{est_asym_svol}, which is a time discretized version of the model described in \cite{cts_time_svol_lev}. The parameters of this model are $\theta = (\phi, \mu, \sigma^2, \rho)$. We write it in a slightly less traditional way so that the notation mirrors that of the previous section \ref{subsec:ssms}: 

\begin{align}
y_t &\mid x_t \sim \mathcal{N}(0, \exp[x_t]) \\
x_1 &\mid \theta \sim \mathcal{N}(\mu, \sigma^2/(1-\phi^2)) \\
x_{t} &\mid y_{t-1}, x_{t-1}, \theta \sim \mathcal{N}(\mu + \phi(x_{t-1} - \mu) + \rho\sigma \exp[-x_{t-1}/2]y_{t-1}, \sigma^2(1-  \rho^2) ).
\end{align}

We assume the parameters are constrained as follows: $\mu \in \mathbb{R}$, $\sigma^2 \in \mathbb{R}^+$, $-1 < \phi < 1$, and $-1 < \rho < 1$.

\section{Algorithms}
\label{sec:algos}



\subsection{Standard Particle Filters} \label{subsec:particle_filtering}

For state space models with a parameter vector that is assumed to be known, particle filtering algorithms provide recursive, simulation-based approximations for expectations with respect to the filtering distributions $p(x_t \mid y_{1:t}, \theta)$ and likelihood evaluations \cite{tutorial}, \cite{gordonetal}, \cite{blake1997condensation}\footnote{One of the important considerations of our model is that the state transition will depend on the previous time's observed quantity. Many descriptions of particle filters will assume that the latent state sequence transition distributions depend only on previous state values. Thankfully, there only needs to be minor modifications in order to handle this \cite{Xu2019}.}. Algorithm \ref{alg:sisr} in Appendix \ref{sec:appendixa} details one step of the Sequential Importance Sampling with Resampling (SISR) algorithm. There are many resampling strategies \cite{resampling_methods} and strategies for choosing different proposal distributions, so this is a fairly generic algorithm. At each time point $t$, we make use of the weighted samples $\{\tilde{X}_t^i, w(\tilde{X}_t^i)\}_{i=1}^N$ to approximate expectations with respect to the filtering distribution, as well as conditional likelihoods $p(y_t \mid y_{1:t-1}, \theta)$. 


When one-step-ahead forecasting is the goal, one will be interested in the first two moments of the forecast distribution (with a parameter that is assumed to be known): $\mathbb{E}[y_{t+1} \mid y_{1:t}, \theta]$ and $\mathbb{V}[y_{t+1} \mid y_{1:t}, \theta]$. For this particular model, the first expectation is defined to be $0$, so we are interested in the second one. As a reminder, particle filters provide approximations to expectations with respect to the filtering distributions. The bridge between these two distributions is provided by iterating expectations. It is easy to show that

\begin{equation}
\mathbf{E}[y^2_{t+1} \mid y_{1:t}, \theta] = \mathbf{E}[ \mu + \phi(x_{t} - \mu) + \rho\sigma \exp[-x_{t}/2]y_{t} + .5\sigma^2(1-  \rho^2) \mid  y_{1:t}, \theta].
\end{equation}
The dependence of this quantity on $y_t$ is sometimes referred to as the {\it leverage effect} in finance \cite{Yu2005}. Instead of approximating this expectation, for the sake of simplicity, we will examine the filter means in section \ref{sec:comp_filt_means}.

We need to reemphasize that these algorithms assume $\theta$ is known. For our case study, when we use the SISR filter for forecasting, we will first approximate the posterior conditioning on a training set, $p(\theta \mid y_{1:s})$. Then we will take a point estimate from this (e.g. $\hat{\theta}_s := \mathbf{E}[\theta \mid y_{1:s}]$) and plug this in to obtain naive ppd approximations: $\hat{p}(y_{t} \mid y_{1:t-1}, \hat{\theta}_s).$


\subsection{Pseudo-Marginal MCMC} \label{subsec:mcmc}

When used for forecasting, the SISR algorithm (algorithm \ref{alg:sisr}) will require a parameter estimate in place of a known parameter. Additionally, the initialization of the particle swarm filter (algorithm \ref{alg:pswarm_filter}) will require posterior samples for its initialization. There are many ways to estimate a posterior--sample-based \cite{brooks2011handbook} and otherwise--but we will exclusively focus our attention on the pseudo-marginal approach \cite{andrieu2009}. Specifically, we will make use of the particle marginal Metropolis-Hastings (PMMH) algorithm \cite{pmcmc}. Given a fixed window of data, this will provide samples distributed according to the parameter posterior $p(\theta \mid y_{1:t})$. 

The primary benefit of this algorithm is that it does not require evaluations of the likelihood. Algorithmically, PMMH is equivalent to the Metropolis-Hastings (MH) algorithm with the exception that an approximate likelihood is used in place of an exact evaluation. Despite using approximate evaluations from a particle filter, under suitable regularity conditions, the Markov chain we produce with PMMH, $\theta^1, \theta^2, \ldots, \theta^N$, targets a posterior of interest exactly. A complete description of this algorithm is provided by algorithm \ref{alg:pmmh} in Appendix \ref{sec:appendixa}.

Note that PMMH requires the user to make choices about which the length of a data window $y_{1:t}$ to condition on, the parameter proposal $q(\cdot \mid \theta)$ used to propose parameters to jump to, and the specific particle filtering algorithm used to approximate the likelihood. The computational efficiency of this algorithm depends heavily on all three of these things. If one can tune this algorithm effectively, though, there is no need for model-specific derivations like there would be in, for instance, Gibbs sampling. More details on how we tuned our PMMH algorithm can be found in Appendix \ref{sec:appendixB}.

\subsection{Two Liu-West Filters}

The original Liu-West Filter \cite{liuandwest} comes from a long line of algorithms that ambitiously attempt to sample from $p(x_t, \theta \mid y_{1:t})$ at every time point. The fundamental idea here is to run one particle filter on a model with the extended state vectors $\tilde{x}_t = (x_t, \theta)$. This was first discussed in \cite{selforganizing}. This paper also describes the primary difficulty with this approach. With a deterministic transition for the parameter component, the unique sample values cannot increase through time. The resampling step of the filtering algorithm further impoverishes the samples. Theoretically, the asymptotic variances of the parameter posteriors increase in time \cite{chopin_clt}.

\cite{selforganizing} also proposed a strategy of mitigation that has been adopted by many papers since. They assume, artificially, that the parameters are no longer static, but in fact are time-varying with known dynamics. In its simplest form, this simply changes the problem of parameter estimation into the problem of filtering on a (hopefully) similar model. In other words, they propose a trouble-free inference task on an approximate model in place of a difficult inference on the exact model. 

This approach is also used in the Liu-West Filter \cite{liuandwest}, which is an algorithm used extensively in section \ref{sec:case_study}. This paper proposes a more complicated set of dynamics for the time-varying ``parameters" in the particle filter in an attempt to lessen the information reduction caused by the introduction of ``artificial evolution noise" for the parameters. The details of this algorithm can be found in algorithm \ref{alg:lwfilter1}. Note that, in addition to suggesting alternative ``parameter" dynamics, the filter described in \cite{liuandwest} resembles the auxiliary particle filter \cite{auxpf}, a competitor to the SISR algorithm we use in this paper. 

Unlike SISR and the particle swarm filter, the Liu-West filters provide both forecasts and parameter posterior samples, and because they only sweep through the data set once, are many orders of magnitude faster than any MCMC sampler. However, this speed comes at a price. As was already mentioned, they do perform filtering on a different model, so they introduce a non quantifiable ``bias" in their estimators \cite{estimationreview}. 

In an attempt to separate the effect of each of these two attributes, we propose and investigate a second version of the Liu-West filter (algorithm \ref{alg:lwfilter2}). It is similar to the first auxiliary-style version in that it has the same ``parameter" dynamics; however, it samples the $x_t$ components of the state with a SISR-like approach. Computationally speaking, this is less complex, so at the very least, we expect this version 2 algorithm to be slightly faster. Finally, our two versions of the Liu-West filter differ from the original description in \cite{liuandwest} in that the state transitions depend on lagged observations. 

As an aside, another strategy that mitigates the difficulty of parameter impoverishment, an alternative to assuming the parameters are time-varying, is to introduce diversity into the parameter samples by using a Markov kernel at each step of the particle filter \cite{mcmcsuffstatspfs} \cite{resamplemove}, \cite{smcsquared}. Instead of replacing one problematic model with a proxy model, it samples new parameters and state paths from the Markov transition kernel that leaves invariant the current target of all of the samples. The downside to this approach is that the computational cost of ``moving" these samples grows in time for most models. Fortunately, for a subclass of models that admit a particular structure, this is not the case \cite{mcmcsuffstatspfs} \cite{storvik} \cite{particlelearning}--however, the variance of these approximations increases in time. Regardless of the benefits and challenges of these algorithms, we do not investigate this class of algorithms because our model does not admit this type of structure, and because we are interested in online forecasting.

\subsection{The Particle Swarm Filter}\label{sec:pswarm}

Instead of sampling both states and parameters together, at every time point, in a single algorithm, one may choose to treat this as two separate tasks. One may first obtain many parameter samples to approximate the marginal posterior $p(\theta \mid y_{1:s})$, and then somehow use these samples to approximate the conditional likelihoods. Then, when more data has been obtained, and there is a new posterior to approximate, the process begins anew. This is the idea behind the particle swarm filter \cite{pswarm}. 

A full description of the algorithm is provided in algorithm \ref{alg:pswarm_filter}. The basic idea is that many noninteracting particle filters are run, each initialized with a randomly chosen parameter. If the parameters are sampled from the posterior $p(\theta \mid y_{1:s})$, then the forecasts for time $s+1$ are asymptotically exact. For $t > s+1$, the forecasts are targetting an approximate ppd:
\begin{equation}
\int p(y_{t} \mid \theta, y_{1:t-1}) p(\theta \mid y_{1:s}) d\theta \label{eqn:approx_ppd}.
\end{equation}
\cite{pswarm} provides some asymptotic results showing the forecasts are asymptotically normal when targetting \ref{eqn:approx_ppd}. If one is concerned about the discrepancy between this and the true ppd, then, after more data has been accumulated, posterior samples may be intermittently updated by successive runs of an MCMC sampler, time permitting. 

Section \ref{sec:rough_timing} provides timing information for how long our particular MCMC algorithm took to run. Additionally, Theorem \ref{thm:bounded_bias} reassures us that the misspecification bias is bounded with a regular updating scheme. Its proof is simple and its assumptions are typical in the particle filtering literature. The case study in section \ref{sec:case_study}, in an effort to disadvantage the particle swarm filter as much as possible, does not perform any update of the posterior samples. 

\begin{theorem}\label{thm:bounded_bias}
Suppose that for any $\theta$, conditioning on some observed data set $y_{1:t}$,
\begin{enumerate}
\itemsep0em 
    \item $p(y_1 \mid \theta) > 0$,
    \item $p(y_t \mid x_{t-1}, \theta) > 0$ for all $t > 1$ and $x_{t-1}$, and
    \item $\sup_{x_t}g(y_t \mid x_t) < \infty$ for all $t \ge 1$.
\end{enumerate}
Then the bias of the particle swarm filter is bounded on a regular updating scheme. 
\end{theorem}

\begin{proof}
For $t > s$
\begin{align*}
\left|p(\theta \mid y_{1:t}) - p(\theta \mid y_{1:s})\right|
&\le 
p(\theta) \left[ \left| \frac{p( y_{1:t} \mid \theta)}{p(y_{1:t})}\right| + \left| \frac{p( y_{1:s} \mid \theta)}{p(y_{1:s})} \right|\right].
\end{align*}

The right hand side is finite because, for any $r \in \mathbb{N}$, $p(y_{1:r}) > 0$ and $p(y_{1:r} \mid \theta) < \infty$. 
\end{proof}

    
\section{The Case Study}
\label{sec:case_study}

This case study uses the stochastic volatility model described in section \ref{sec:svol_mod} and uses the forecasting algorithms described in section \ref{sec:algos}. The data set, a univariate time series, is comprised of continuously-compounded percent returns of an exchange traded fund that tracks the S\&P 500 stock index (ticker: SPY). The data starts on January 1, 2010, and runs until July 29, 2022. This data is freely available from Yahoo! Finance, and was downloaded using the \verb|quantmod| library in \verb|R| \cite{quantmod}. For the purposes of this study, we will call the training data $y_{1:s}$ and the test data $y_{s+1:T}$.

The code of this case-study is freely available on a public Github repository \cite{Brown_A_Case-Study_of_2022}. The files are organized into three parts: data getting/cleaning, analysis, and visualization. Most of the high-level instructions are written in \verb|R| scripts. The computationally intensive portions--the filters and the MCMC samplers--are written in \verb|c++| and make use of the \verb|pf| and \verb|ssme| libraries \cite{pf_paper} \cite{Brown_SSME_a_c_2022}. \verb|GNU Make| \cite{gnu_make} is used to manage the entire project's workflow \cite{JSSv094c01}. Plots are created with the \verb|ggplot2| library in \verb|R| \cite{ggplot} \cite{ggtern} \cite{ggally}.

The results reported in this paper come from the code being run on a Dell Latitude 5420 laptop running Ubuntu 20.04.4 LTS. The machine has an 11th Gen Intel(R) Core(TM) i5-1145G7 @ 2.60GHz chip and 16 gigabytes of RAM. The \verb|c++| portion of the code was compiled with \verb|-O3| flags using the \verb|g++ 9.4.0| compiler. Double precision floating point numbers were used for all individual scalars as well as for the template parameters of the Eigen matrices \cite{eigenweb}.

Before we discuss any results, let us examine Table \ref{tab:high_level_comparison}, which provides a summary of the pros and cons of each forecasting approach. 

\begin{center}
\begin{table}[]
\begin{tabular}{l|llll}
\cline{2-5}
                                                     & \multicolumn{1}{l|}{Liu-West (vers. 1)}            & \multicolumn{1}{l|}{Liu-West (vers. 2)}            & \multicolumn{1}{l|}{Particle Swarm}                & \multicolumn{1}{l|}{Particle Filter}               \\ \hline
\multicolumn{1}{|l|}{considers $\theta$ uncertainty} & \cellcolor[HTML]{009901}{\color[HTML]{000000} yes} & \cellcolor[HTML]{009901}{\color[HTML]{000000} yes} & \cellcolor[HTML]{009901}{\color[HTML]{000000} yes} & \cellcolor[HTML]{CB0000}{\color[HTML]{000000} no}  \\ \cline{1-1}
\multicolumn{1}{|l|}{provides $\theta$ samples}      & \cellcolor[HTML]{009901}{\color[HTML]{000000} yes} & \cellcolor[HTML]{009901}{\color[HTML]{000000} yes} & \cellcolor[HTML]{009901}{\color[HTML]{000000} yes} & \cellcolor[HTML]{CB0000}{\color[HTML]{000000} no}  \\ \cline{1-1}
\multicolumn{1}{|l|}{targets ppd}                    & \cellcolor[HTML]{009901}{\color[HTML]{000000} yes} & \cellcolor[HTML]{009901}{\color[HTML]{000000} yes} & \cellcolor[HTML]{009901}{\color[HTML]{000000} yes} & \cellcolor[HTML]{CB0000}{\color[HTML]{000000} no}  \\ \cline{1-1}
\multicolumn{1}{|l|}{requires auxiliary MCMC}        & \cellcolor[HTML]{009901}{\color[HTML]{000000} no}  & \cellcolor[HTML]{009901}{\color[HTML]{000000} no}  & \cellcolor[HTML]{CB0000}{\color[HTML]{000000} yes} & \cellcolor[HTML]{CB0000}{\color[HTML]{000000} yes} \\ \cline{1-1}
\multicolumn{1}{|l|}{supporting ppd asymptotics}     & \cellcolor[HTML]{CB0000}{\color[HTML]{000000} no}  & \cellcolor[HTML]{CB0000}{\color[HTML]{000000} no}  & \cellcolor[HTML]{009901}{\color[HTML]{000000} yes} & \cellcolor[HTML]{CB0000}{\color[HTML]{000000} no}  \\ \cline{1-1}
\multicolumn{1}{|l|}{pleasingly parallel}                 & \cellcolor[HTML]{CB0000}{\color[HTML]{000000} no}  & \cellcolor[HTML]{CB0000}{\color[HTML]{000000} no}  & \cellcolor[HTML]{009901}{\color[HTML]{000000} yes} & \cellcolor[HTML]{CB0000}{\color[HTML]{000000} no}  \\ \cline{1-1}
\end{tabular}
\caption{\label{tab:high_level_comparison}A high-level comparison of all forecasting algorithms.}
\end{table}
\end{center}
\subsection{Rough Timing}\label{sec:rough_timing}

Here we give a general sense of how long each algorithm takes to run. The runtimes are displayed in Table \ref{tab:comparison}. We run the following filters across the entire data set $y_{1:T}$ (3164 observations):

\begin{enumerate}[topsep=0pt,itemsep=-1ex,partopsep=1ex,parsep=1ex]
\item SISR algorithm: $100$ state particles, $\theta$ set to posterior mean of MCMC samples
\item Liu-West version 1 (auxiliary style): $500$ state and parameter particles, $\theta$ parameters drawn from MCMC samples, $\delta = .99$
\item Liu-West version 1 (auxiliary style): $500$ state and parameter particles, $\theta$ parameters drawn from independent (and informative) uniform posteriors, $\delta = .99$
\item Liu-West version 2: $500$ state and parameter particles, $\theta$ parameters drawn from MCMC samples, $\delta = .99$
\item Liu-West version 2: $500$ state and parameter particles, $\theta$ parameters drawn from independent (and informative) uniform posteriors, $\delta = .99$
\item Particle Swarm: $100$ state particles, $100$ parameter particles, $\theta$ parameters drawn from MCMC samples, multithreading across 8 cores
\item Particle Swarm: $100$ state particles, $100$ parameter particles, $\theta$ parameters drawn from independent (and informative) uniform posteriors, multithreading across 8 cores
\end{enumerate}

\begin{table}[]
\centering
\begin{tabular}{lll}
\cline{2-2}
\multicolumn{1}{l|}{}                    & \multicolumn{1}{l|}{Elapsed Time (seconds)} &  \\ \cline{1-2}
\multicolumn{1}{|l|}{Liu-West (vers. 1), MCMC samples} & \multicolumn{1}{l|}{9.91}            &  \\ \cline{1-2}
\multicolumn{1}{|l|}{Liu-West (vers. 1), uniform samples} & \multicolumn{1}{l|}{8.46}            &  \\ \cline{1-2}
\multicolumn{1}{|l|}{Liu-West (vers. 2), MCMC samples} & \multicolumn{1}{l|}{4.64}            &  \\ \cline{1-2}
\multicolumn{1}{|l|}{Liu-West (vers. 2), uniform samples} & \multicolumn{1}{l|}{4.60}            &  \\ \cline{1-2}
\multicolumn{1}{|l|}{Particle Swarm, MCMC samples}     & \multicolumn{1}{l|}{2.42}            &  \\ \cline{1-2}
\multicolumn{1}{|l|}{Particle Swarm, uniform samples}     & \multicolumn{1}{l|}{2.90}            &  \\ \cline{1-2}
\multicolumn{1}{|l|}{Particle Filter}    & \multicolumn{1}{l|}{0.56}              &  \\ \cline{1-2}
                                         &                                       &  \\
                                         &                                       &  \\
                                         &                                       & 
\end{tabular}
\caption{\label{tab:comparison}Elapsed time on all filtering algorithms.}
\end{table}
Both Liu-West filters and the particle swarm filter are run with different posterior sampling strategies. One strategy samples directly from a hard-coded distribution, while the other draws from posterior samples stored in a file. The difference in runtimes gives some sense of how long it takes to read from that external file. The support of the marginal uniform posteriors was taken directly from the MCMC samples--the left and right bounds correspond with 95\% credible intervals. More information about MCMC priors, tuning and diagnostics can be found in Appendix \ref{sec:appendixB}. 

In addition to these filtering programs, which sweep through the data once, we also ran 100,000 iterations of the pseudo-marginal MCMC sampler (Algorithm \ref{alg:pmmh}) on the training data $y_{1:s}$. After all, as was mentioned in sections \ref{subsec:particle_filtering} and \ref{sec:pswarm}, the particle swarm filter and the standard SISR filter need to be initialized with posterior estimates or samples. This sampling took 1.1 hours.

Each iteration of this MCMC sampler runs through the training data set once using 7 SISR filters (Algorithm \ref{alg:sisr}) and averages their likelihoods together. The particle filters are run in parallel, and this averaging increases the efficiency of the algorithm. Averaging doesn't just attentuate noise--it preserves the unbiasedness of the likelihood, which is a core requirement of the pseudo-marginal approach \cite{andrieu2009}.

\subsection{Comparison of Posterior Samples}\label{sec:posterior_samp_comp}

Next, a comparison of posterior samples is provided. For this experiment, we use the training data $y_{1:s}$ (251 observations) containing of all the returns from the year 2010. We run three programs for this portion of the case study:
\begin{enumerate}[topsep=0pt,itemsep=-1ex,partopsep=1ex,parsep=1ex]
\item Pseudo-marginal MCMC (algorithm \ref{alg:pmmh}): $7$ bootstrap particle filters run in parallel at each iteration, each using $100$ state particles (same as section \ref{sec:rough_timing})
\item Liu-West vers. 1 (auxiliary style): $500$ state and parameter particles, $\delta = .99$
\item Liu-West vers. 2: $500$ state and parameter particles, $\delta = .99$
\end{enumerate}
We use the same prior in all approaches--the same one used in \cite{Yu2005}. All four parameters are assumed to be independent, a priori, and 
\begin{enumerate}
    \item $\pi([\phi+1]/2) = \text{Beta}(20,1.5)$
    \item $\pi(\mu) = \text{Normal}(0, 25)$
    \item $\pi(\sigma^2) = \text{Inverse-Gamma}(2.5, .025)$
    \item $\pi(\rho) = \text{Uniform}(-1,1)$.
\end{enumerate}

The Liu-West filters are each run $100$ times. Boxplots for the marginal posteriors are displayed in Figures \ref{fig:many_posts_1} and \ref{fig:many_posts_2}. It is no surprise that, given the difference in runtimes and theoretical guarantees, that the parameter samples of the Liu-West filters are quite inferior to the ones obtained by the PMMH algorithm. 


\begin{figure}[h]
     \centering
     \begin{subfigure}[b]{0.45\textwidth}
         \centering
         \includegraphics[width=\textwidth]{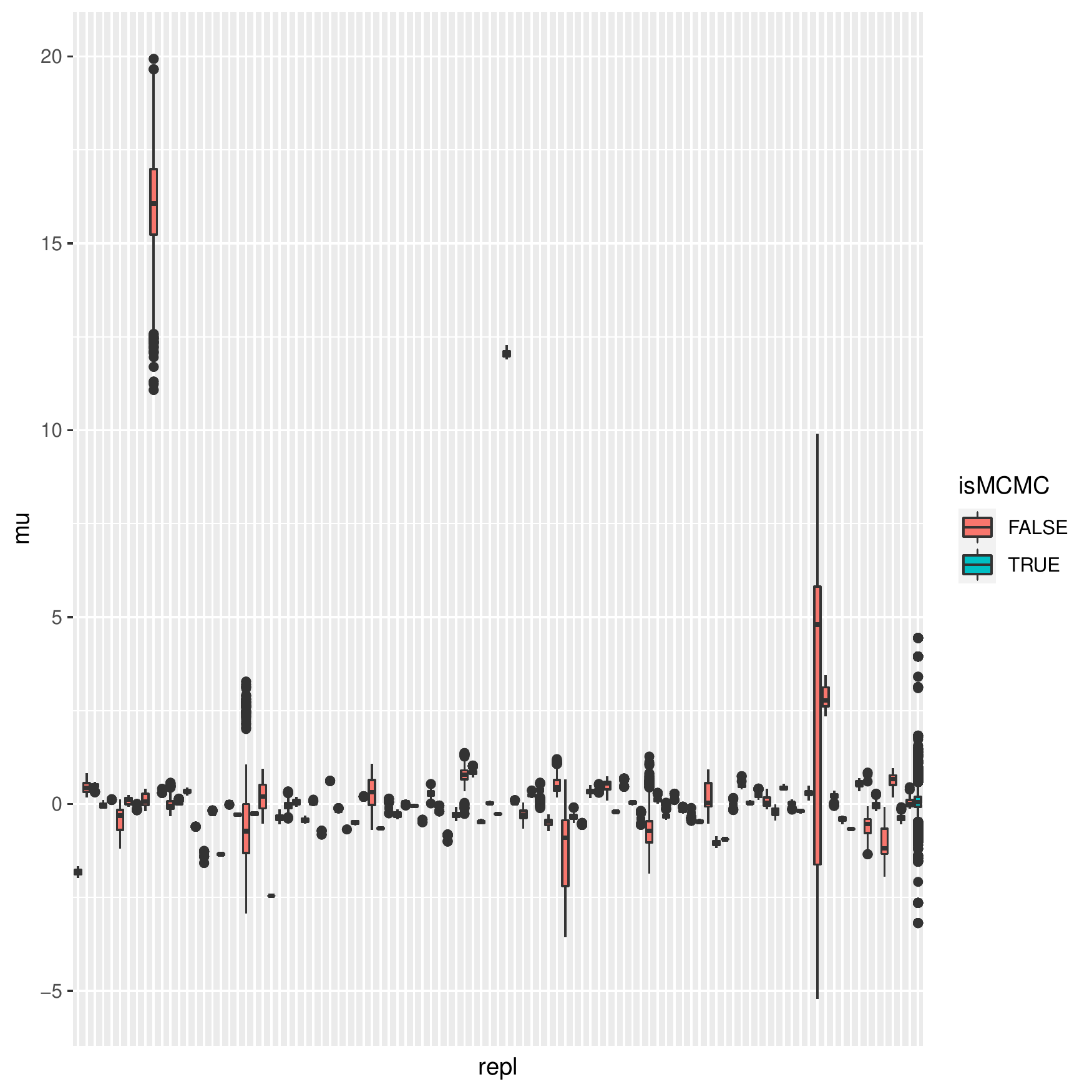}
         \caption{$100$ replicates $\hat{p}(\mu \mid y_{1:s})$ samples}
         \label{fig:many_mu_posts_1}
     \end{subfigure}
     \hfill
     \begin{subfigure}[b]{0.45\textwidth}
         \centering
         \includegraphics[width=\textwidth]{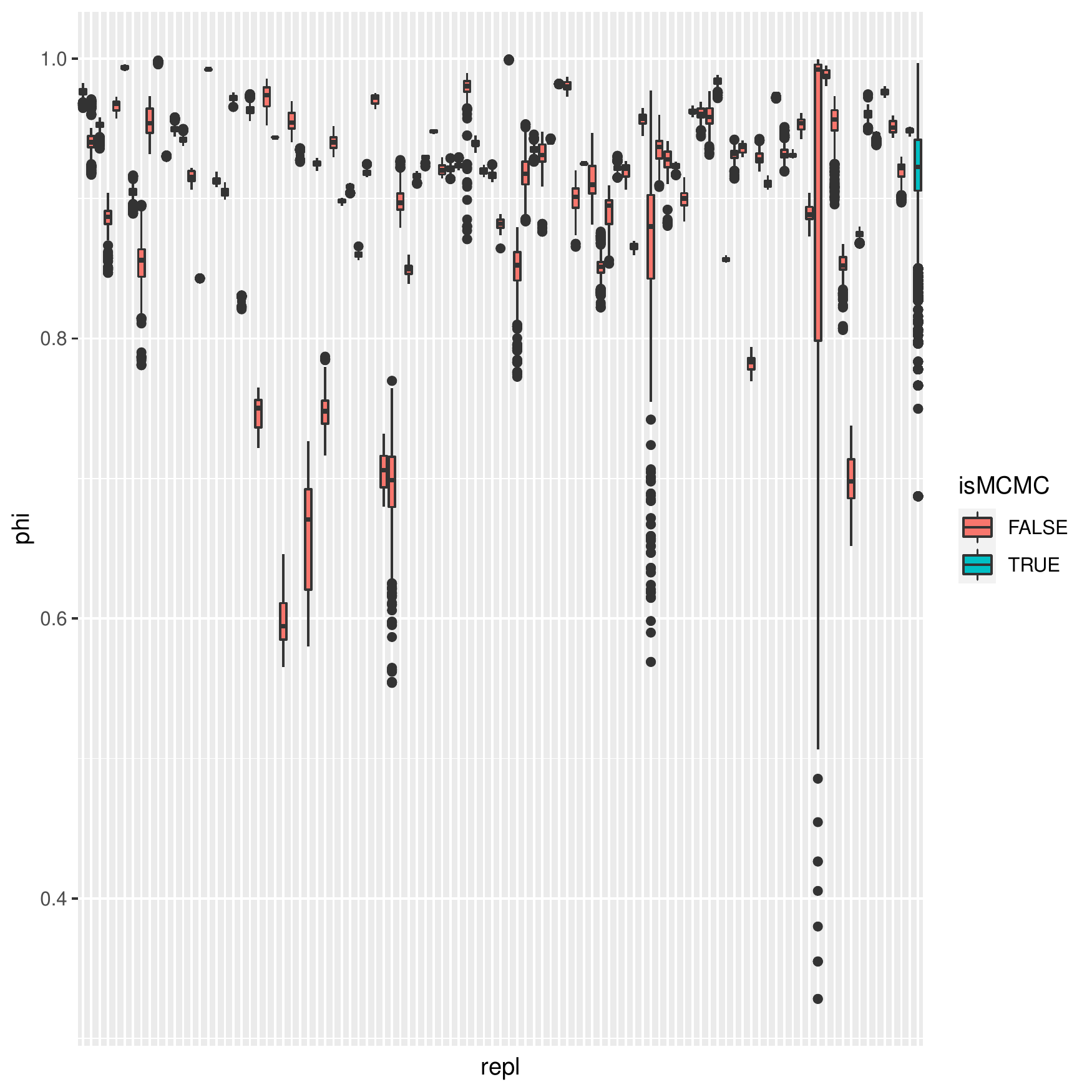}
         \caption{$100$ replicates $\hat{p}(\phi \mid y_{1:s})$ samples}
         \label{fig:many_phi_posts_1}
     \end{subfigure}
     \begin{subfigure}[b]{0.45\textwidth}
         \centering
         \includegraphics[width=\textwidth]{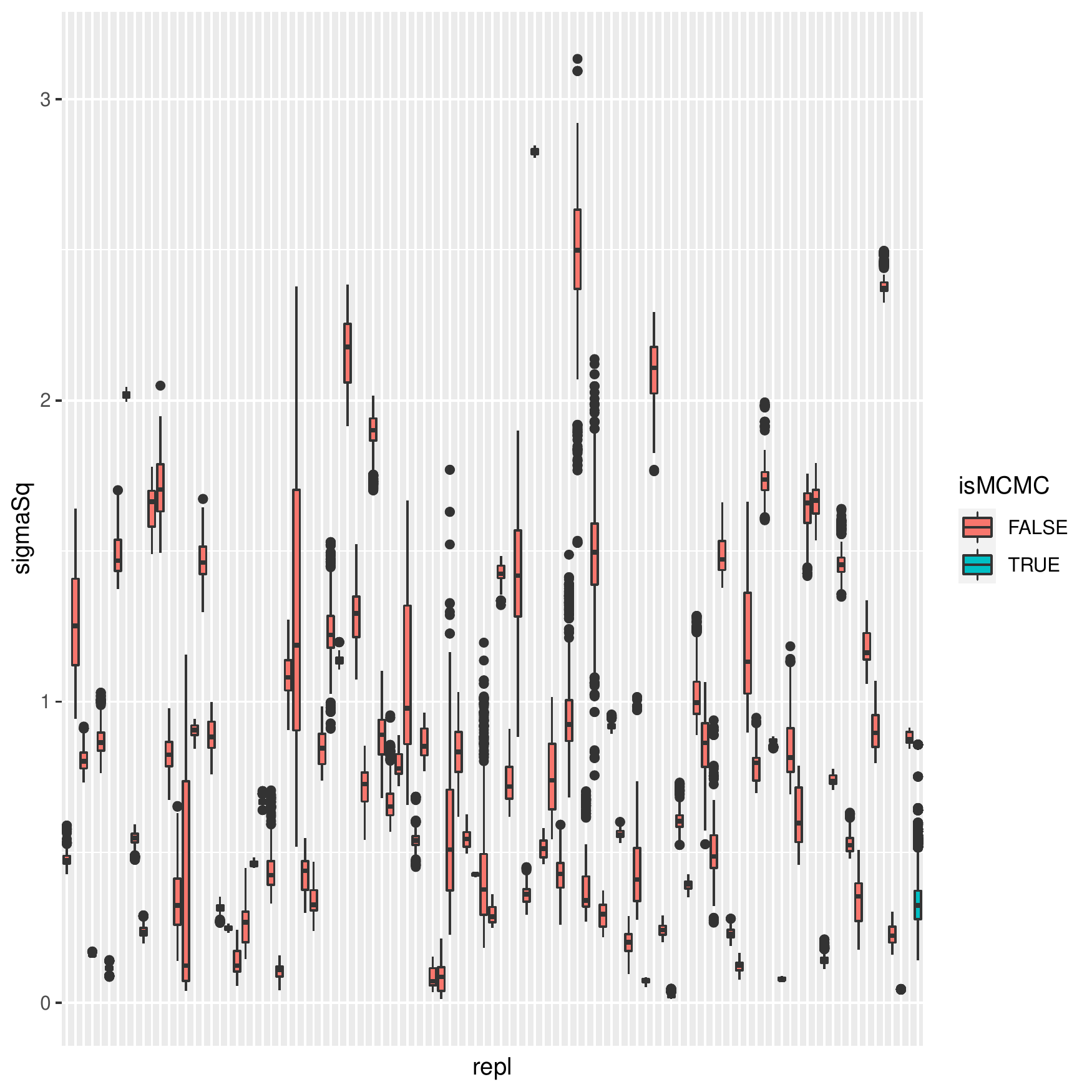}
         \caption{$100$ replicates $\hat{p}(\sigma^2 \mid y_{1:s})$ samples}
         \label{fig:many_sigmaSq_posts_1}
     \end{subfigure}
     \hfill
     \begin{subfigure}[b]{0.45\textwidth}
         \centering
         \includegraphics[width=\textwidth]{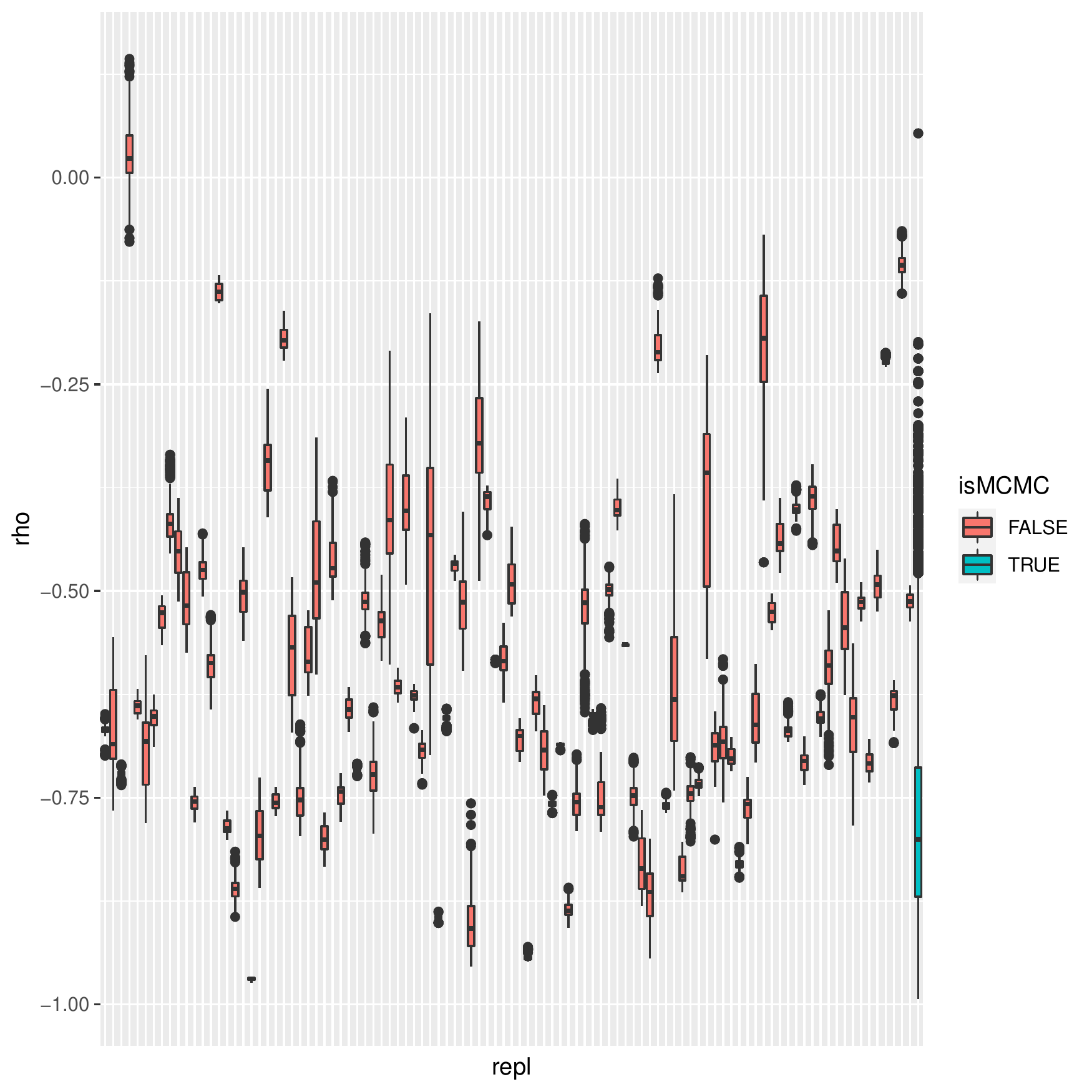}
         \caption{$100$ replicates $\hat{p}(\rho \mid y_{1:s})$ samples}
         \label{fig:many_rho_posts_1}
     \end{subfigure}
        \caption{Marginal posterior approximations from version $1$ of the Liu-West filter compared with MCMC samples.}
        \label{fig:many_posts_1}
\end{figure}

\begin{figure}[h]
     \centering
     \begin{subfigure}[b]{0.45\textwidth}
         \centering
         \includegraphics[width=\textwidth]{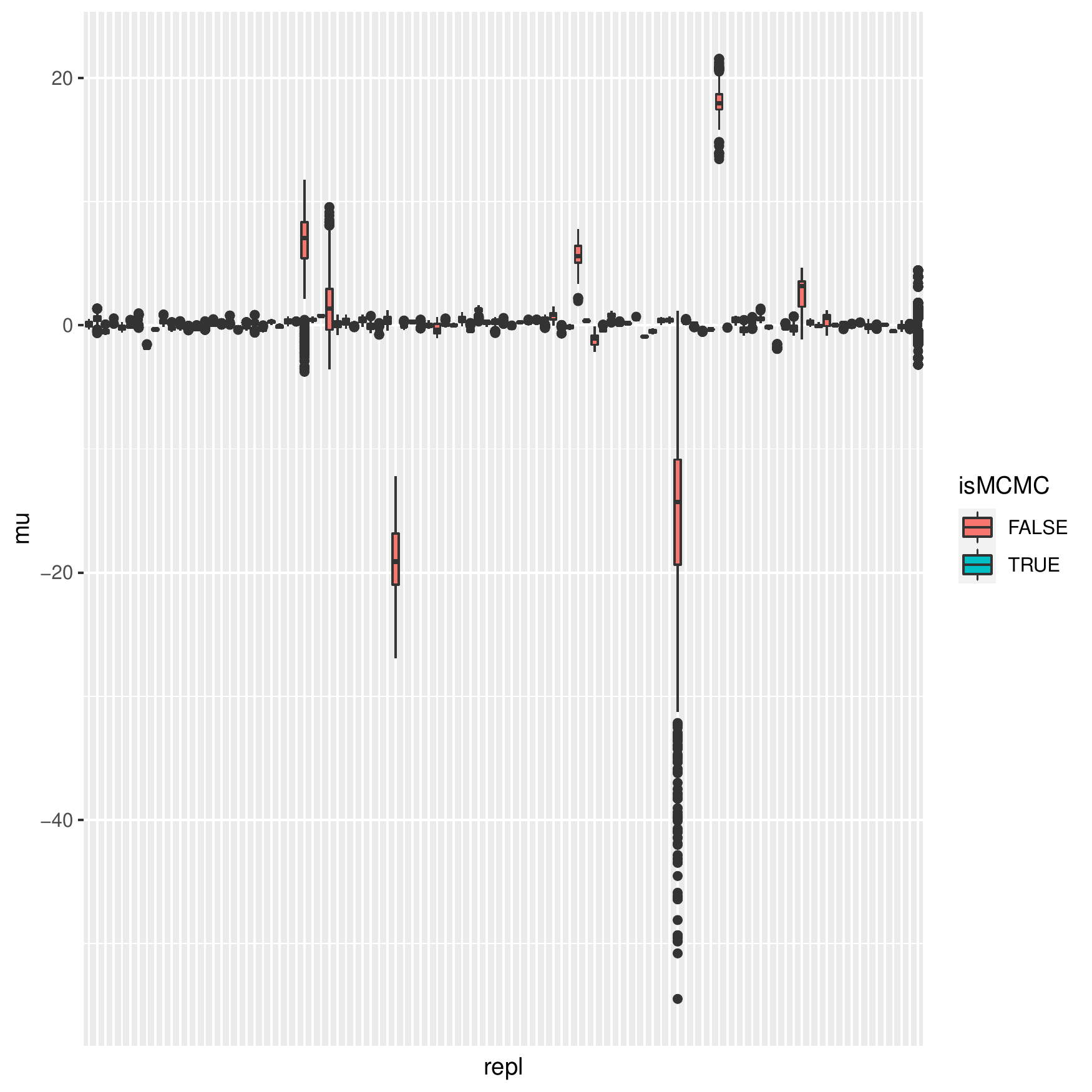}
         \caption{$100$ replicates $\hat{p}(\mu \mid y_{1:s})$ samples}
         \label{fig:many_mu_posts_2}
     \end{subfigure}
     \hfill
     \begin{subfigure}[b]{0.45\textwidth}
         \centering
         \includegraphics[width=\textwidth]{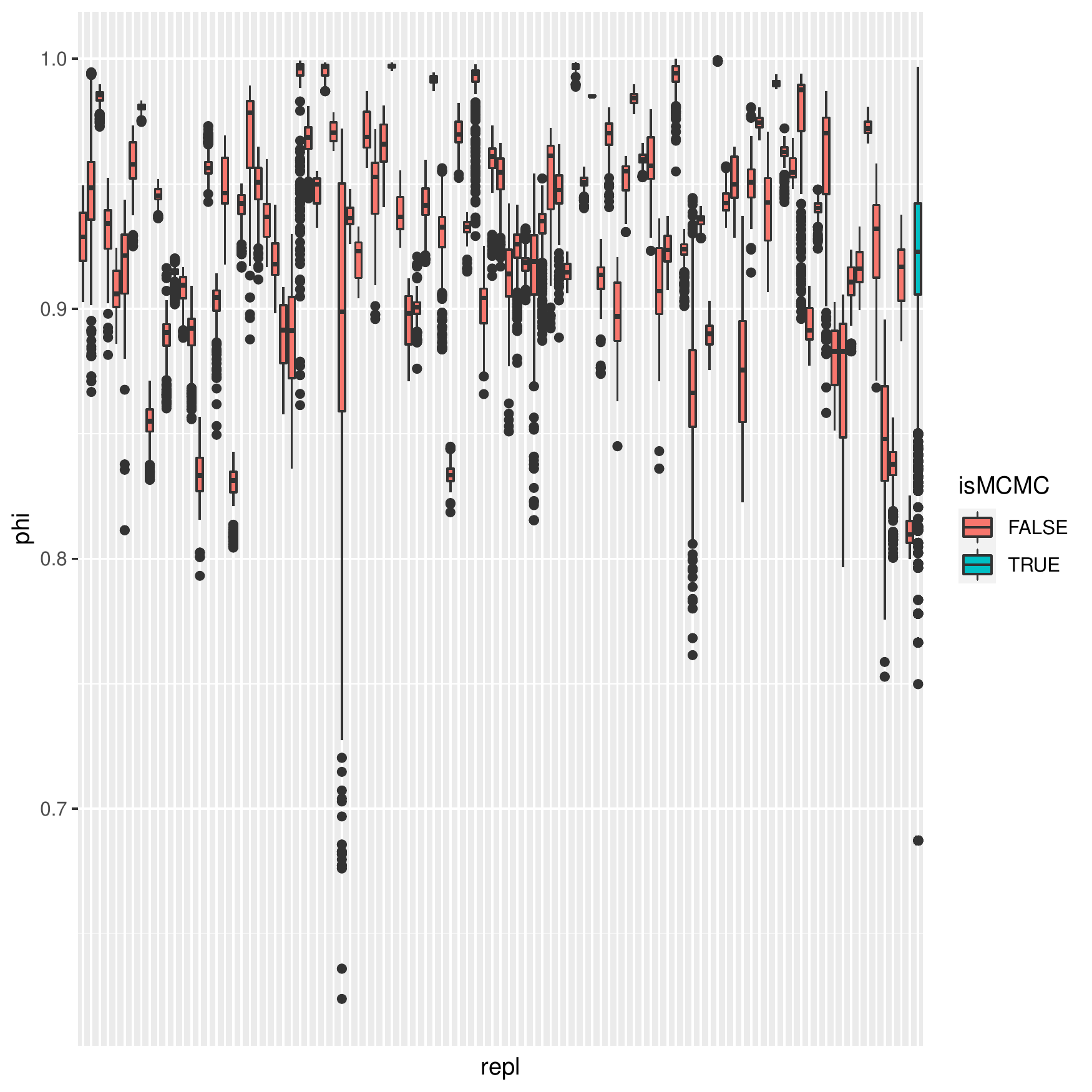}
         \caption{$100$ replicates $\hat{p}(\phi \mid y_{1:s})$ samples}
         \label{fig:many_phi_posts_2}
     \end{subfigure}
     \begin{subfigure}[b]{0.45\textwidth}
         \centering
         \includegraphics[width=\textwidth]{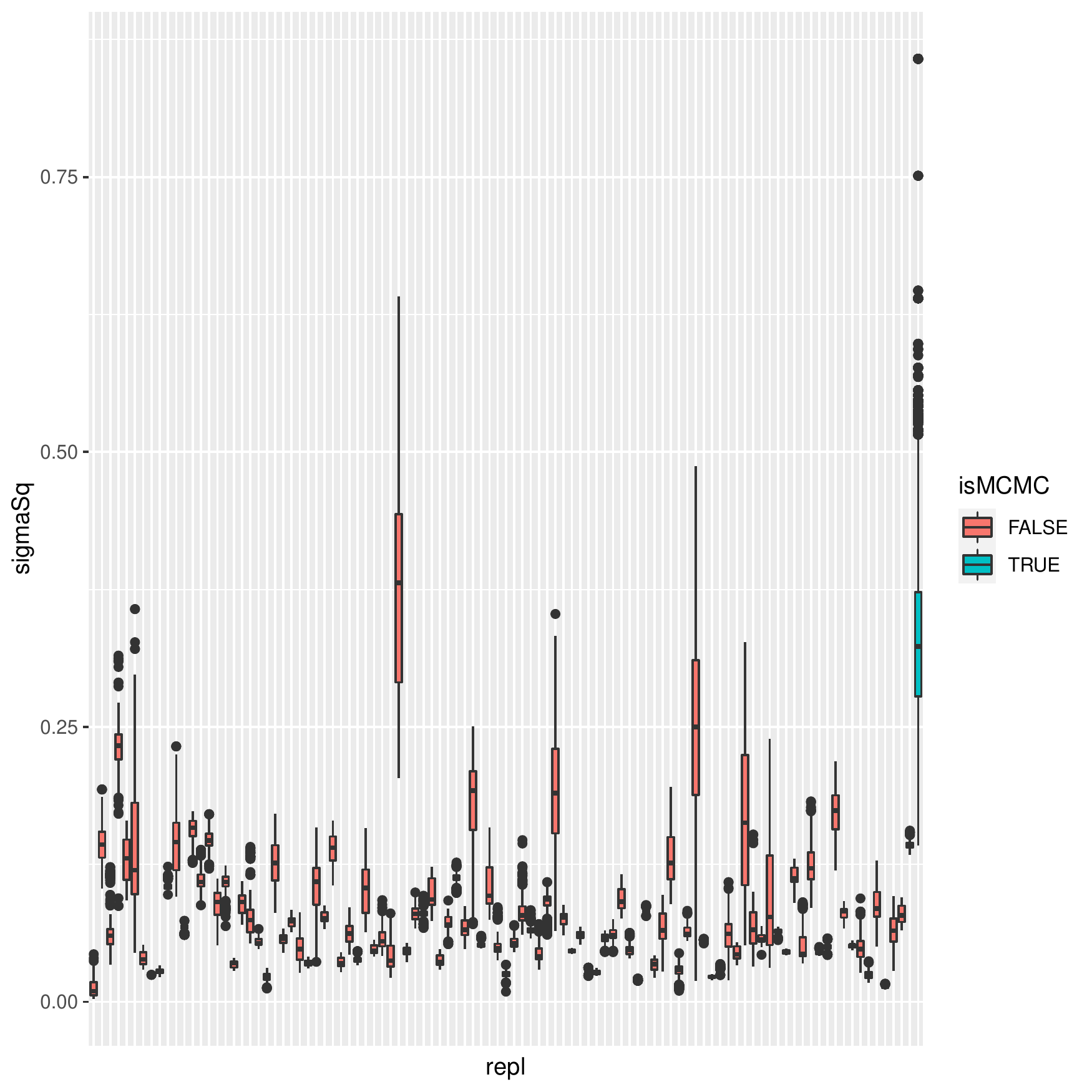}
         \caption{$100$ replicates $\hat{p}(\sigma^2 \mid y_{1:s})$ samples}
         \label{fig:many_sigma_posts_2}
     \end{subfigure}
     \hfill
     \begin{subfigure}[b]{0.45\textwidth}
         \centering
         \includegraphics[width=\textwidth]{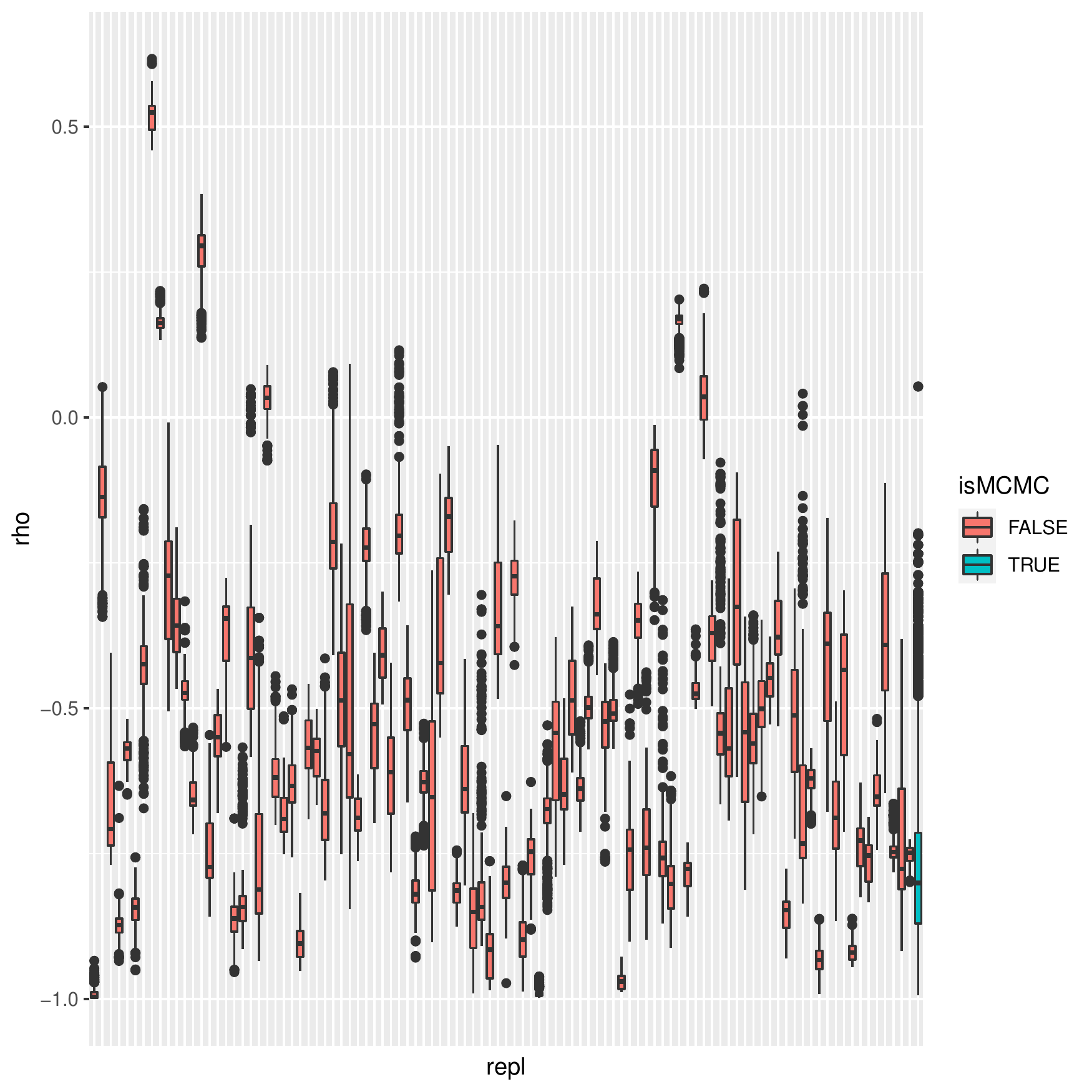}
         \caption{$100$ replicates $\hat{p}(\rho \mid y_{1:s})$ samples}
         \label{fig:many_rho_posts_2}
     \end{subfigure}
        \caption{Marginal posterior approximations from version $2$ of the Liu-West filter compared with MCMC samples.}
        \label{fig:many_posts_2}
\end{figure}


\subsection{Comparison of Filter Means}\label{sec:comp_filt_means}

This section provides a comparison of our forecast volatilities, also known as the estimates of the filtered mean: $\mathbb{E}(x_{t} \mid y_{1:t})$. The filters are initialized with draws from the same starting posteriors used in section \ref{sec:rough_timing}.

All of the approaches seem to yield estimates that move in tandem, and all of the filtered mean sequences seem to increase when the dispersion of returns increases (see Figure \ref{fig:filter_mean_and_returns}). 

As a reminder, all three algorithms are ``biased," but for different reasons. The SISR filter does not take into account uncertainty, the particle swarm filter uses an outdated parameter posterior, and the Liu-West filters are filtering for different models. Despite all of these differences, all three approaches seem to produce equivalent volatility estimates.

\begin{figure}
     \centering
     \begin{subfigure}[b]{0.45\textwidth}
         \centering
         \includegraphics[width=\textwidth]{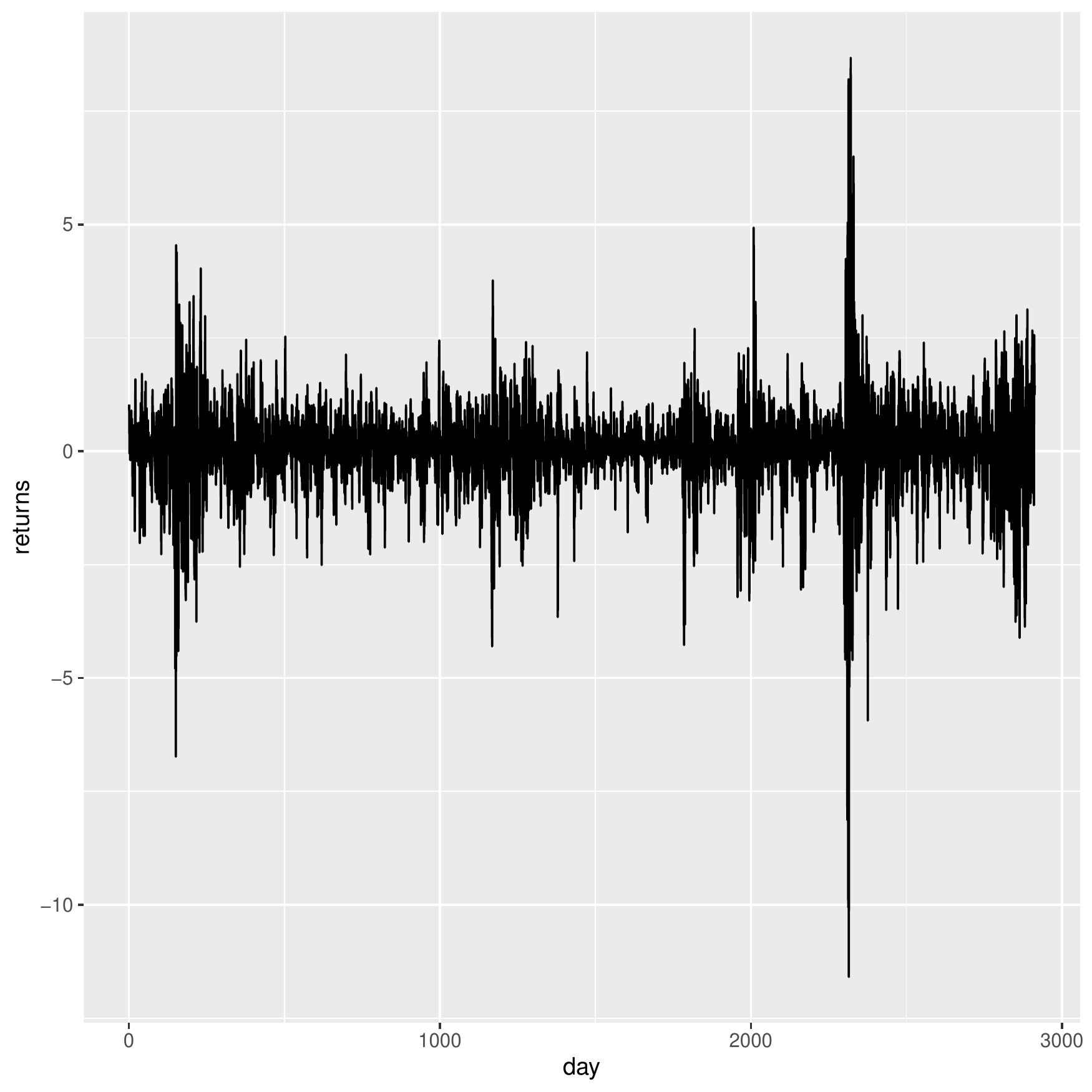}
         \caption{observed returns in test set $y_{s+1:T}$}
         \label{fig:obs_returns}
     \end{subfigure}
     \hfill
     \begin{subfigure}[b]{0.45\textwidth}
         \centering
         \includegraphics[width=\textwidth]{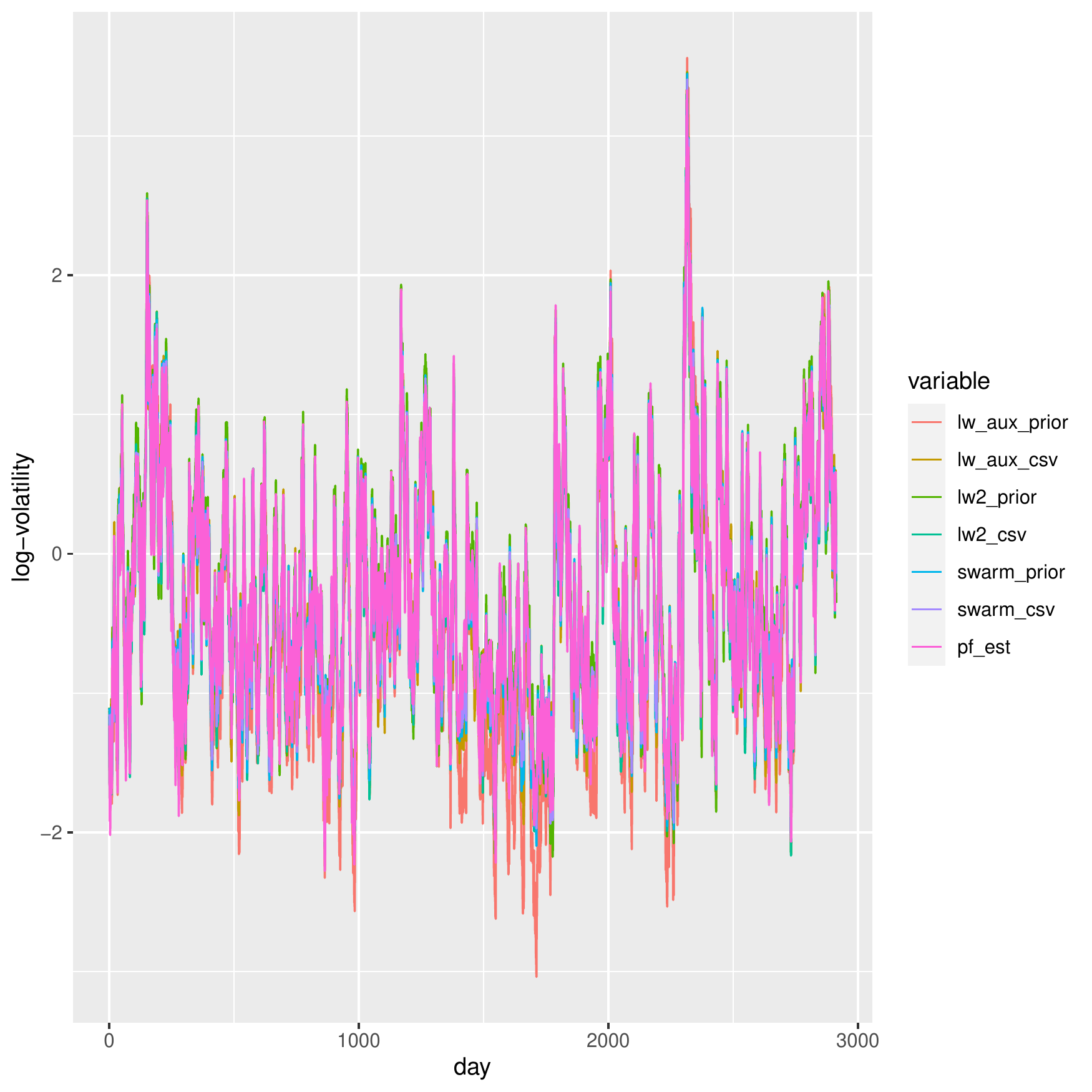}
         \caption{$\hat{\mathbb{E}}(x_t \mid y_{1:t})$}
         \label{fig:filt_means}
     \end{subfigure}
        \caption{Approximate ``volatility" for all algorithms.}
        \label{fig:filter_mean_and_returns}
\end{figure}

\subsection{Comparison of Evidence Approximations}

A holistic way to evaluate forecasts is to use a \emph{scoring rule}, which measures the value of a forecast by looking at its entire distribution, instead of just its moments. The value is computed by comparing the distribution against the observed return that comes to pass \cite{scoring_rules}. The logarithmic scoring rule $p(y_t \mid y_{1:t-1})$, otherwise known as the conditional evidence, is arguably the most important because of its familiarity and mathematical properties \cite{log_scoring_rle}. Unfortunately, for the class of models that we are considering, only approximate evaluations are available for this function at each point in time. 

Theoretically, each algorithm is being run on the same data set, with the same model, using the same priors. If the approximations were perfect, all of these quantities would be equivalent. This section investigates any discrepancies between approximations. The approximations to the (log of the) conditional evidences are shown in Figure \ref{fig:log_clikes}. The cumulative sums of these are shown in Figure \ref{fig:log_clikes_cumsum}.

Recall that the particle swarm filter provides for the opportunity to renew parameter samples. However, in an attempt to disadvantage the particle filter and the particle swarm filter as much as possible, their parameter posterior samples are never updated--only the first set of parameter samples obtained from the training data are used for the life of the experiment. This decision is also made in the description in Algorithm \ref{alg:pswarm_filter}. The particle swarm filter makes use of these raw parameter posterior samples, while the SISR particle filter is instantiated with the posterior mean vector. 

\begin{figure}
     \centering
     \begin{subfigure}[b]{0.45\textwidth}
         \centering
         \includegraphics[width=\textwidth]{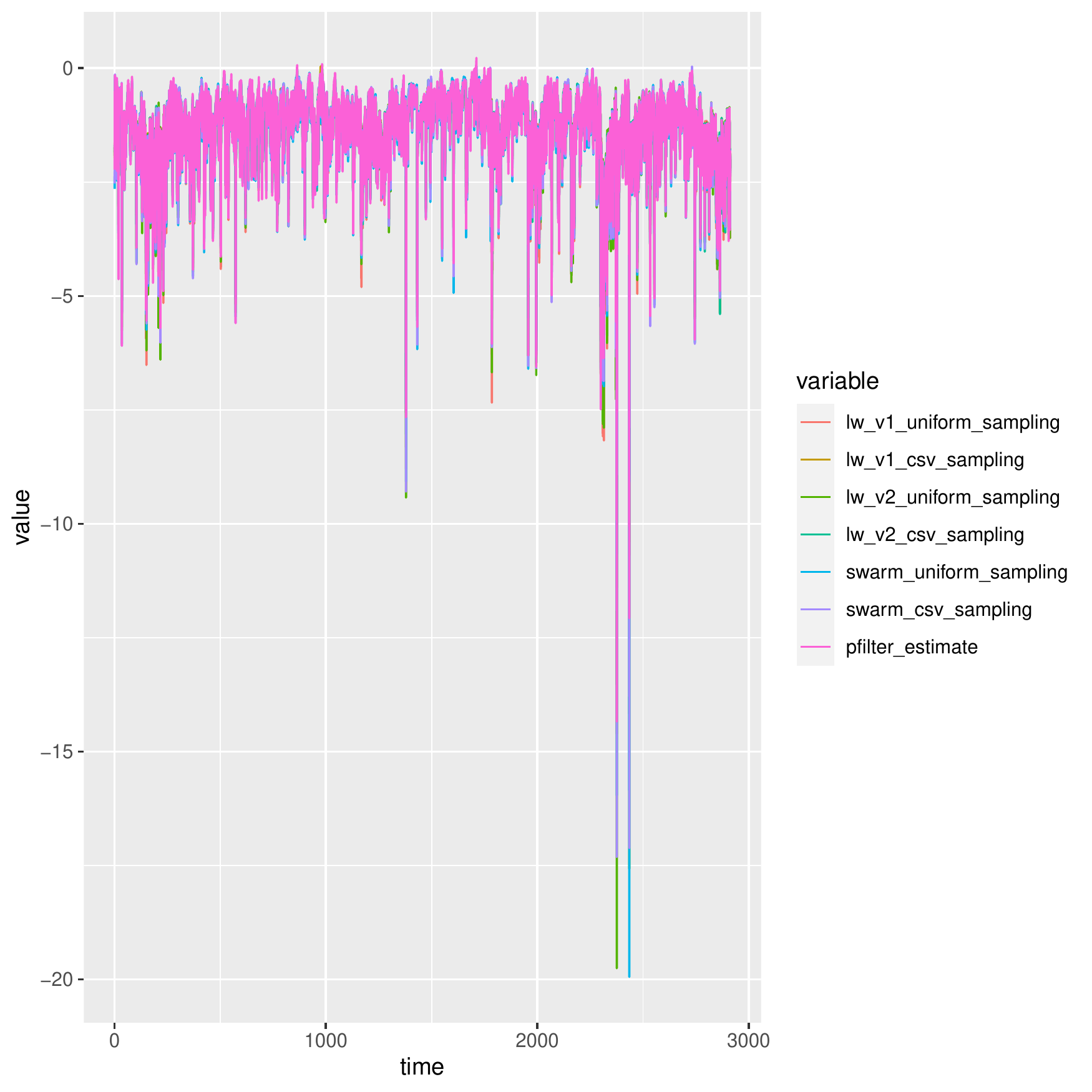}
         \caption{$\log \hat{p}(y_t \mid y_{1:t-1})$ for $t > s$}
         \label{fig:log_clikes}
     \end{subfigure}
     \hfill
     \begin{subfigure}[b]{0.45\textwidth}
         \centering
         \includegraphics[width=\textwidth]{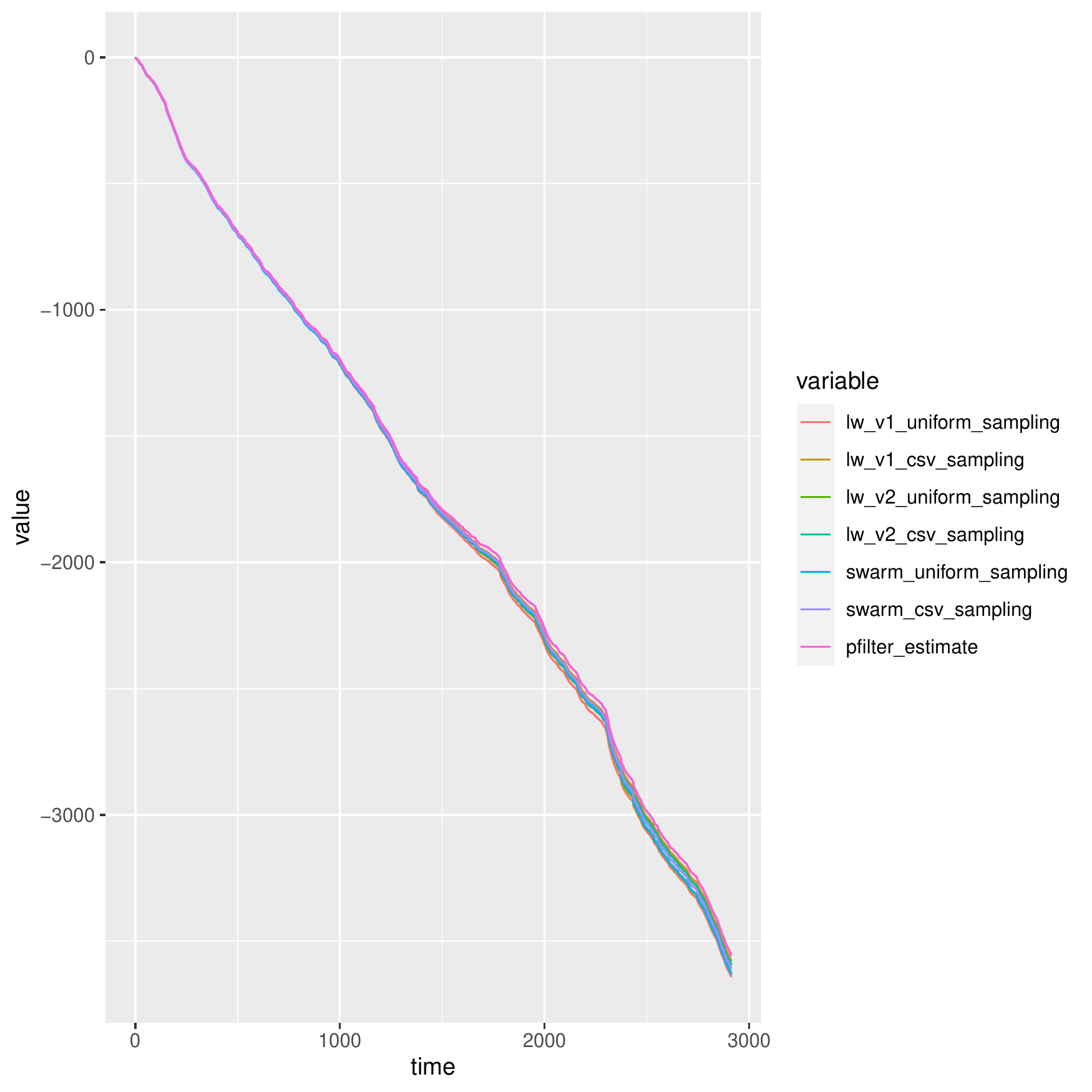}
         \caption{$\log \hat{p}(y_{s+1:t} \mid y_{1:s})$ for $t > s$}
         \label{fig:log_clikes_cumsum}
     \end{subfigure}
        \caption{Approximate log conditional evidences for all algorithms.}
        \label{fig:three graphs}
\end{figure}

Next, we visualize the relative size of each conditional evidence provided by each algorithm. Compare the two Liu-West filters and the particle swarm filter with each other. For each algorithm, at each time, take its conditional evidence, and divide by the sum of all three conditional evidences. We can do this twice--one for comparing the algorithms when sampling is performed from a text file, and one for comparing the algorithms when sampling is performed from a hardcoded distribution. 

Define $\mathbf{r}_t$ and $\mathbf{s}_t$ to each be vector-valued time series evolving on the (3-1)-simplex with the $i$th element proportional to $\hat{p}_i(y_t \mid y_{1:t-1})$. Figure \ref{fig:relative_cond_likes} plots these two time series in ternary plots. These plots show that all three algorithms tend to produce equivalent values throughout the entire time sequence.

\begin{figure}
     \centering
     \begin{subfigure}[b]{0.45\textwidth}
         \centering
         \includegraphics[width=\textwidth]{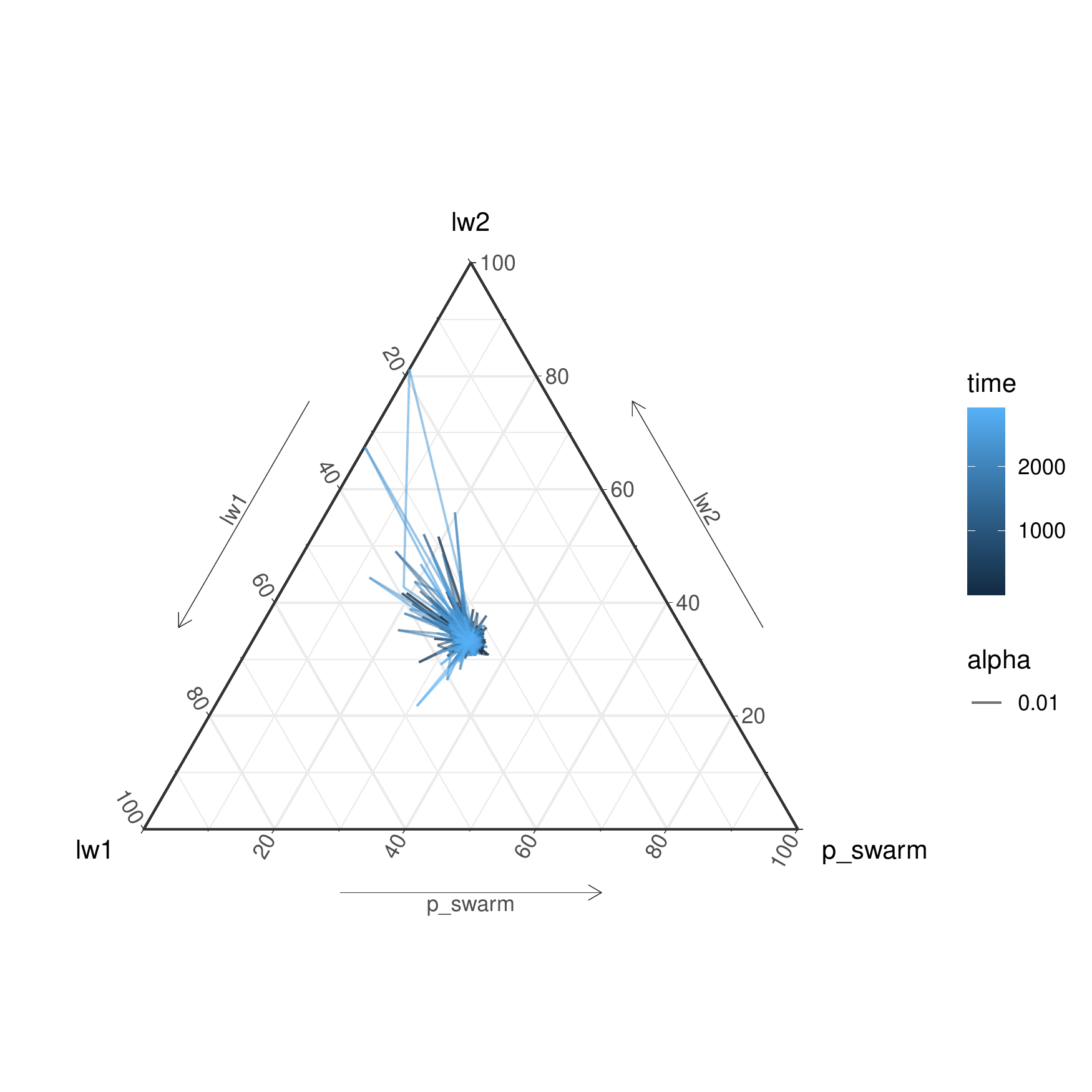}
         \caption{$\mathbf{r}_t$ for csv sampling}
         \label{fig:rel_csv_cond_likes}
     \end{subfigure}
     \hfill
     \begin{subfigure}[b]{0.45\textwidth}
         \centering
         \includegraphics[width=\textwidth]{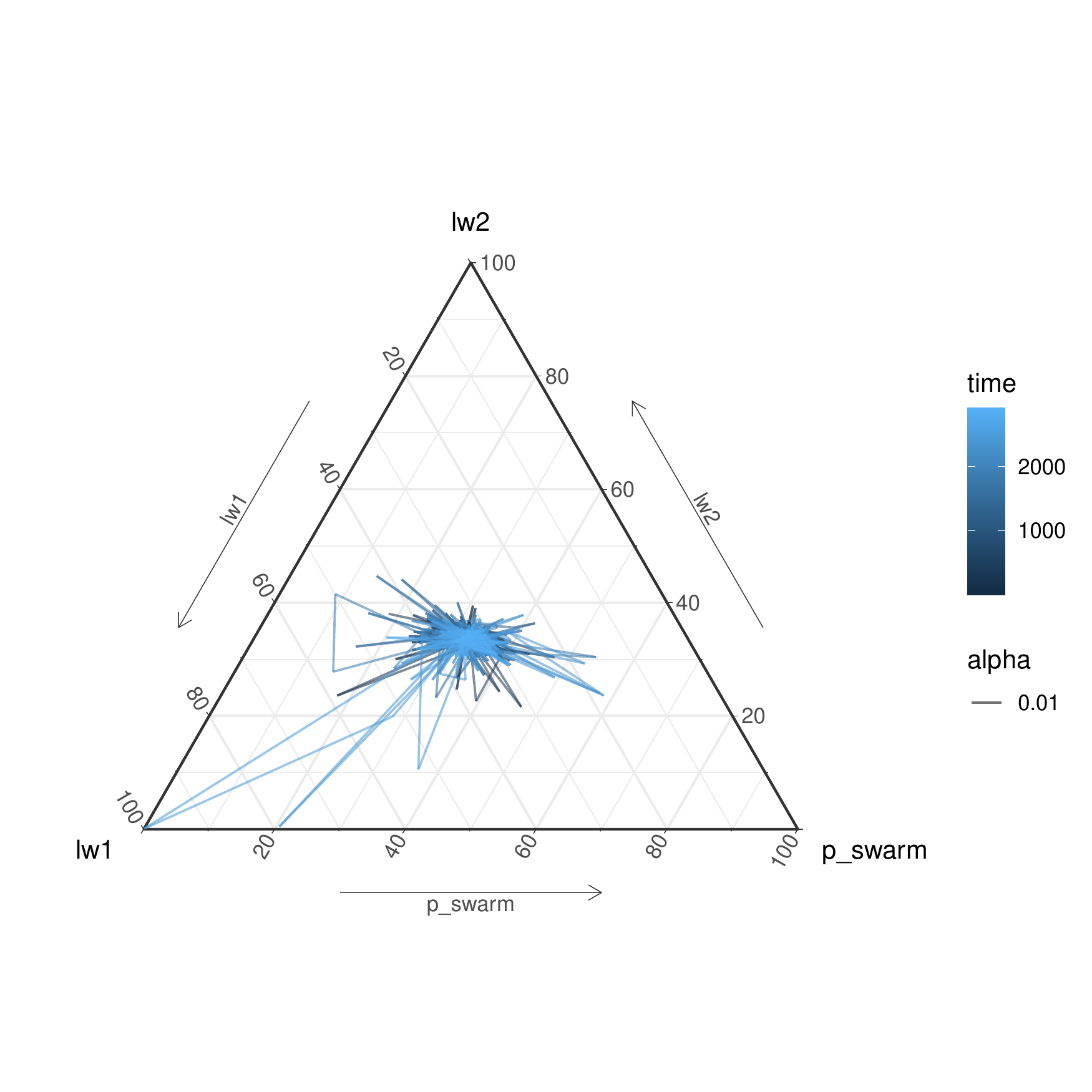}
         \caption{$\mathbf{s}_t$ for uniform sampling}
         \label{fig:rel_uniform_cond_likes}
     \end{subfigure}
        \caption{Comparing relative size of conditional likelihoods.}
        \label{fig:relative_cond_likes}
\end{figure}

\section{Conclusion}
\label{sec:conc}

In this paper, we explored the output generated by several forecasting algorithms. We considered ourselves to be Bayesian and we restricted our interest to sampling-based approaches that were suitable for real-time applications. We initially suspected that any differences might be explained by either how parameter uncertainty is taken into account, or by the presence of artificial parameter evolution noise. Despite these differences, we find that all the algorithms we considered in this paper produce similar forecasts. In addition, we find that the parameter posterior samples from the Liu-West filters are a poor substitute for those coming from a Markov chain Monte Carlo algorithm. 


\begin{appendices}

\section{Appendix A: Description of Some Algorithms}\label{sec:appendixa}

\begin{algorithm}
\caption{PMMH}\label{alg:pmmh}
{\fontsize{10}{4}\selectfont
\begin{algorithmic}[0]
\Procedure{pmmh}{}
  \For{$i=1,\ldots, N$}
    \If{$i$ equals $1$}
      \State Choose or sample $\theta^1$ such that $p(y_{1:t}, \theta^1) > 0$    
      \State Store $\hat{p}(y_{1:t} \mid \theta^1)$ 
    \Else\Comment{$i > 1$}
        \State Sample $\theta' \sim q(\theta' \mid \theta^{i-1})$
        \State Sample $U_t \sim \text{Uniform}(0,1)$
        \State Compute $a(\theta,\theta') := \frac{p(\theta')\hat{p}(y_{1:t} \mid \theta') q(\theta^{i-1} \mid \theta') }{p(\theta^{i-1})\hat{p}(y_{1:t} \mid \theta^{i-1}) q(\theta' \mid \theta^{i-1}) }$
        \If{$U_t < \min(1, a(\theta,\theta'))$}
            \State Set $\theta^i = \theta'$
            \State Set $\hat{p}(y_{1:t} \mid \theta^i) = \hat{p}(y_{1:t} \mid \theta')$
        \Else{}
            \State Set $\theta^i = \theta^{i-1}$
            \State Set $\hat{p}(y_{1:t} \mid \theta^i) = \hat{p}(y_{1:t} \mid \theta^{i-1})$
        \EndIf
    \EndIf
  \EndFor
\EndProcedure
\end{algorithmic}
}
\end{algorithm}


\begin{algorithm}
\caption{SISR}\label{alg:sisr}
{\fontsize{10}{4}\selectfont
\begin{algorithmic}[0]
\Procedure{sisr (at time $t$)}{}
  \If{$t$ equals $1$}
    \For{$i=1,\ldots, N$}
      \State Sample $\tilde{X}_1^i \sim q(x_1 \mid y_1, \theta)$  
      \State Compute $w(\tilde{X}_1^i) = \frac{g(y_1 \mid \tilde{X}^i_1)\mu(\tilde{X}_1^i \mid \theta)}{q(\tilde{X}_1^i \mid y_1, \theta)}$
    \EndFor
    \State Compute $\hat{p}(y_1) = N^{-1} \sum_{i=1}^N \tilde{w}(\tilde{X}_1^i)$ 
    \State Compute $\hat{\mathbb{E}}[f(x_1) \mid y_{1},\theta] = \sum_{i=1}^N \frac{w(\tilde{X}_1^i)}{ \sum_k w(\tilde{X}_1^k) } f(\tilde{X}_1^i)$ for any integrable $f$
    \If{resampling criterion satisfied}
    \For{$j=1,\ldots, N$} 
      \State Sample $I_1^j \sim \text{Cat}\left(
      \frac{w(\tilde{X}_1^1)}{ \sum_k w(\tilde{X}_1^k) }, 
      \ldots, 
      \frac{w(\tilde{X}_1^N)}{ \sum_k w(\tilde{X}_1^k) }
      \right) $ 
      \State Assign $X_1^j = \tilde{X}_1^{I_1^j}$
      \State Assign $w(X_1^j) = 1$
    \EndFor
    \Else \Comment{see note}
    \For{$j=1,\ldots, N$} 
      \State Assign $X_1^j = \tilde{X}_1^{j}$
      \State Assign $w(X_1^j) = w(\tilde{X}_1^j)$
    \EndFor
    \EndIf
  \Else\Comment{$t > 1$}
  \For{$i=1,\ldots N$}
    \State Sample $\tilde{X}_t^i \sim q(x_t \mid y_t, X_{t-1}^i, \theta)$ 
    \State Compute $w(\tilde{X}_t^i) = w(X_{t-1}^i) \frac{g(y_t \mid \tilde{X}^i_t)f(\tilde{X}_t^i \mid X^i_{t-1}, y_{t-1}, \theta)}{q(\tilde{X}_t^i \mid y_t, X_{t-1}^i, \theta)}$
  \EndFor
  \State Compute $\hat{p}(y_t \mid y_{1:t-1}) = N^{-1} \sum_{i=1}^N
  w(\tilde{X}_t^i)$ 
    \State Compute $\hat{\mathbb{E}}[f(x_t) \mid y_{1:t},\theta] = \sum_{i=1}^N \frac{w(\tilde{X}_t^i)}{ \sum_k w(\tilde{X}_t^k) } f(\tilde{X}_t^i)$for any integrable $f$ 
    \If{resampling criterion satisfied}
    \For{$j=1,\ldots, N$} 
      \State Sample $I_t^j \sim \text{Cat}\left(
      \frac{w(\tilde{X}_t^1)}{ \sum_k w(\tilde{X}_t^k) }, 
      \ldots, 
      \frac{w(\tilde{X}_t^N)}{ \sum_k w(\tilde{X}_t^k) }
      \right) $ 
      \State Assign $X_t^j := \tilde{X}_t^{I_t^j}$
    \EndFor
    \Else \Comment{see note}
    \For{$j=1,\ldots, N$} 
      \State Assign $X_t^j = \tilde{X}_t^{j}$
      \State Assign $w(X_t^j) = w(\tilde{X}_t^j)$
    \EndFor
    \EndIf
  \EndIf
\EndProcedure
\end{algorithmic}
}
\end{algorithm}



\begin{algorithm}
\caption{Liu-West Filter (auxiliary style and at time $t$)}\label{alg:lwfilter1}
{\fontsize{10}{4}\selectfont
\begin{algorithmic}[0]
\Procedure{auxiliary style Liu-West Filter }{}
    \If{$t=1$} 
        \State Store $a = (3\delta -1)/(2\delta) $ where $\delta$ is a user-chosen tuning parameter.
        \State Select $s$ and obtain posterior samples from $p(\theta \mid y_{1:s})$ ($s = 0$ corresponds with no posterior simulation)
        \For{$i=1,\ldots,N$} 
            \State Sample $\tilde{\theta}^i_1 \sim p(\theta \mid y_{1:s})$
            \State Sample $\tilde{X}_1^{i} \sim q(x_1 \mid y_1, \tilde{\theta}_1^i)$
            \State Compute $w(\tilde{X}_1^{i}, \tilde{\theta}_1^i) = \frac{g(y_1 \mid \tilde{X}_1^{i}, \tilde{\theta}^i_1)\mu(\tilde{X}_1^{i} \mid \tilde{\theta}^i_1)}{q(\tilde{X}_1^{i} \mid y_1, \tilde{\theta}^i_1)}$
        \EndFor
        \State Compute $\hat{p}(y_1) = N^{-1} \sum_{i} w(\tilde{X}_1^i, \tilde{\theta}_1^i)$
        \State Compute $\hat{\mathbb{E}}[f(x_1) \mid y_{1}] = \sum_{i=1}^N \frac{w(\tilde{X}_1^i, \tilde{\theta}_1^i)}{ \sum_j w(\tilde{X}_1^j, \tilde{\theta}_1^j) } f(\tilde{X}_1^i)$
        \If{resampling criterion satisfied}
        \For{$i=1,\ldots,N$}
                \State Sample $I_1^i$ with $P(I_1^i = k) \propto w(\tilde{X}_1^{k}, \tilde{\theta}_1^k)$, $k \in \{ 1, \ldots, N\}$
                \State Assign $X_1^i = \tilde{X}_1^{I_1^i}$ and $\theta_1^i = \tilde{\theta}_1^{I_1^i}$
                \State Assign $w(X_1^{i}, \theta_1^i) = 1$
        \EndFor
        \Else \Comment{see note}
        \For{$i=1,\ldots,N$}
            \State Assign $X_1^i = \tilde{X}_1^{i}$ and $\theta_1^i = \tilde{\theta}_1^{i}$
            \State Assign $w(X_1^{i}, \theta_1^i) = w(\tilde{X}_1^{i}, \tilde{\theta}_1^i) $
        \EndFor
        \EndIf
    \Else \Comment{ ($t > 1$)}
        \For{$i=1,\ldots,N$} 
            \State Compute $\mu_t^i = \mathbb{E}(x_t \mid x_{t-1}^i, y_{t-1}, \theta_{t-1}^i)$  
            \State Compute $m^i_{t-1} = a \theta_{t-1}^i + (1-a) N^{-1} \sum_{j} \theta_{t-1}^j$
        \EndFor
        \For{$i=1,\ldots,N$} 
            \State Sample $I_t^i$ with $P(I_t^i = k) \propto w(X_{t-1}^k, \theta_{t-1}^k)g(y_t \mid \mu_t^k, m_{t-1}^k)$, $k \in \{ 1, \ldots, N\}$
            \State Sample $\tilde{\theta}_t^i \sim \mathcal{N}(m_{t-1}^{I_t^i}, h^2 V_{t-1})$ where $V_t$ is the sample covariance matrix of the parameter particles.
            \State Sample $\tilde{X}_t^i \sim f(x_t \mid X_{t-1}^{I_t^i}, y_{t-1}, \tilde{\theta}_t^i)$
            \State Compute $w(\tilde{X}_t^i, \tilde{\theta}_t^i) = \frac{g(y_t \mid \tilde{X}_t^i, \tilde{\theta}_t^i) }{g(y_t \mid \mu_t^{I_t^i}, m_{t-1}^{I_t^i})}$
        \EndFor
        \State Compute $\hat{p}(y_t \mid y_{1:t-1}) = N^{-1} \sum_{i} w(\tilde{X}_t^i, \tilde{\theta}_t^i)$
        \State Compute $\hat{\mathbb{E}}[f(x_t) \mid y_{1:t}] = \sum_{i=1}^N \frac{w(\tilde{X}_t^i, \tilde{\theta}_t^i)}{ \sum_j w(\tilde{X}_t^j, \tilde{\theta}_t^j) } f(\tilde{X}_t^i)$
        \For{$i=1,\ldots,N$}
            \State derp
        \EndFor
        \If{resampling criterion satisfied}
        \For{$i=1,\ldots,N$}
                \State Sample $I_t^i$ with $P(I_t^i = k) \propto w(\tilde{X}_t^{k}, \tilde{\theta}_t^k)$, $k \in \{ 1, \ldots, N\}$
                \State Assign $X_t^i = \tilde{X}_t^{I_t^i}$ and $\theta_t^i = \tilde{\theta}_t^{I_t^i}$
                \State Assign $w(X_t^{i}, \theta_t^i) = 1$
        \EndFor
        \Else \Comment{see note}
        \For{$i=1,\ldots,N$}
            \State Assign $X_t^i = \tilde{X}_t^{i}$ and $\theta_t^i = \tilde{\theta}_t^{i}$
            \State Assign $w(X_t^{i}, \theta_t^i) = w(\tilde{X}_t^{i}, \tilde{\theta}_t^i) $
        \EndFor
        \EndIf
    \EndIf
\EndProcedure
\end{algorithmic}
}
\end{algorithm}

\begin{algorithm}
\caption{Liu-West Filter (alternative style and at time $t$)}\label{alg:lwfilter2}
{\fontsize{10}{4}\selectfont
\begin{algorithmic}[0]
\Procedure{alternative style Liu-West Filter }{}
    \If{$t=1$} 
        \State Store $a = (3\delta -1)/(2\delta) $ where $\delta$ is a user-chosen tuning parameter.
        \State Select $s$ and obtain posterior samples from $p(\theta \mid y_{1:s})$ ($s = 0$ corresponds with no posterior simulation)
        \For{$i=1,\ldots,N$} 
            \State Sample $\tilde{\theta}^i_1 \sim p(\theta \mid y_{1:s})$
            \State Sample $\tilde{X}_1^{i} \sim q(x_1 \mid y_1, \tilde{\theta}_1^i)$
            \State Compute $w(\tilde{X}_1^{i}, \tilde{\theta}_1^i) = \frac{g(y_1 \mid \tilde{X}_1^{i}, \tilde{\theta}^i_1)\mu(\tilde{X}_1^{i} \mid \tilde{\theta}^i_1)}{q(\tilde{X}_1^{i} \mid y_1, \tilde{\theta}^i_1)}$
        \EndFor
        \State Compute $\hat{p}(y_1) = N^{-1} \sum_{i} w(\tilde{X}_1^i, \tilde{\theta}_1^i)$
        \State Compute $\hat{\mathbb{E}}[f(x_1) \mid y_{1}] = \sum_{i=1}^N \frac{w(\tilde{X}_1^i, \tilde{\theta}_1^i)}{ \sum_j w(\tilde{X}_1^j, \tilde{\theta}_1^j) } f(\tilde{X}_1^i)$
        \If{resampling criterion satisfied}
        \For{$i=1,\ldots,N$}
                \State Sample $I_1^i$ with $P(I_1^i = k) \propto w(\tilde{X}_1^{k}, \tilde{\theta}_1^k)$, $k \in \{ 1, \ldots, N\}$
                \State Assign $X_1^i = \tilde{X}_1^{I_1^i}$ and $\theta_1^i = \tilde{\theta}_1^{I_1^i}$
                \State Assign $w(X_1^{i}, \theta_1^i) = 1$
        \EndFor
        \Else \Comment{see note}
        \For{$i=1,\ldots,N$}
            \State Assign $X_1^i = \tilde{X}_1^{i}$ and $\theta_1^i = \tilde{\theta}_1^{i}$
            \State Assign $w(X_1^{i}, \theta_1^i) = w(\tilde{X}_1^{i}, \tilde{\theta}_1^i) $
        \EndFor
        \EndIf
    \Else \Comment{ ($t > 1$)}
        \For{$i=1,\ldots,N$} 
            \State Compute $m^i_{t-1} = a \theta_{t-1}^i + (1-a) N^{-1} \sum_{j} \theta_{t-1}^j$
            \State Sample $\tilde{\theta}_t^i \sim \mathcal{N}(m_{t-1}^{i}, h^2 V_{t-1})$ where $V_t$ is the sample covariance matrix of the parameter particles.   
            \State Sample $\tilde{X}_t^i \sim f(x_t \mid X_{t-1}^{i}, y_{t-1}, \tilde{\theta}_t^i)$
            \State Compute $w(\tilde{X}_t^i, \tilde{\theta}_t^i) = w(X_{t-1}^i, \theta_{t-1}^i) \frac{g(y_t \mid \tilde{X}_t^i, \tilde{\theta}_t^i)f(\tilde{X}_t^i \mid X_{t-1}^i, y_{t-1}, \tilde{\theta}_t^i) }{q(\tilde{X}_t^i \mid X_{t-1}^i, y_t, \tilde{\theta}_t^i)}$
        \EndFor
        \State Compute $\hat{p}(y_t \mid y_{1:t-1}) = N^{-1} \sum_{i} w(\tilde{X}_t^i, \tilde{\theta}_t^i)$
        \State Compute $\hat{\mathbb{E}}[f(x_t) \mid y_{1:t}] = \sum_{i=1}^N \frac{w(\tilde{X}_t^i, \tilde{\theta}_t^i)}{ \sum_j w(\tilde{X}_t^j, \tilde{\theta}_t^j) } f(\tilde{X}_t^i)$
        \If{resampling criterion satisfied}
        \For{$i=1,\ldots,N$}
                \State Sample $I_t^i$ with $P(I_t^i = k) \propto w(\tilde{X}_t^k, \tilde{\theta}_t^k)$, $k \in \{ 1, \ldots, N\}$
                \State Assign $X_t^i = \tilde{X}_t^{I_t^i}$ and $\theta_t^i = \tilde{\theta}_t^{I_t^i}$
                \State Assign $w(X_t^{i}, \theta_t^i) = 1$
        \EndFor
        \Else \Comment{see note}
        \For{$i=1,\ldots,N$}
            \State Assign $X_t^i = \tilde{X}_t^{i}$ and $\theta_t^i = \tilde{\theta}_t^{i}$
            \State Assign $w(X_t^{i}, \theta_t^i) = w(\tilde{X}_t^{i}, \tilde{\theta}_t^i) $
        \EndFor
        \EndIf
    \EndIf
\EndProcedure
\end{algorithmic}
}
\end{algorithm}


\begin{algorithm}
\caption{Particle Swarm Filter}\label{alg:pswarm_filter}
{\fontsize{10}{4}\selectfont
\begin{algorithmic}[0]
\Procedure{Particle Swarm Filter }{}
\For{$i=1,\ldots,N_\theta$}
    \State Draw $\theta^i \sim p(\theta \mid y_{1:s})$ ($s = 0$ corresponds with prior simulation)
\EndFor
\For{$t=1,2,\ldots$}
    \If{$t=1$}
        \For{$i=1,\ldots,N_\theta$} 
            \For{$j=(i-1)N + 1, \ldots, i N$} 
                \State Sample $\tilde{X}_1^{j} \sim q(x_1 \mid y_1, \theta^i)$
                \State Compute $w(\tilde{X}_1^{j}) := \frac{g(y_1 \mid \tilde{X}_1^{j})\mu(\tilde{X}_1^{j} \mid \theta^i)}{q(\tilde{X}_1^{j} \mid y_1, \theta^i)}$
            \EndFor
            \State Compute $\hat{p}(y_1 \mid \theta^i) := N^{-1} \sum_{j} w(\tilde{X}_1^j)$
            \State Compute $\hat{\mathbb{E}}[f(x_1) \mid y_{1},\theta^i] = \sum_{j=1}^N \frac{w(\tilde{X}_1^j)}{ \sum_k w(\tilde{X}_1^k) } f(\tilde{X}_1^j)$ 
            \For{$j=(i-1)N + 1, \ldots, i N$}
                \State Sample $I_1^j$ with $P(I_1^j = k) \propto w(\tilde{X}_1^{k})$, $k \in \{ (i-1)N + 1, \ldots, i N\}$
                \State Assign $X_1^j := \tilde{X}_1^{I_1^j}$
            \EndFor
        \EndFor
        \State Compute $\hat{p}(y_1) := N_\theta^{-1} \sum_{i} \hat{p}(y_1 \mid \theta^i)$
    \Else \Comment{ ($t > 1$)}
        \For{$i=1,\ldots,N_\theta$}
            \For{$j=(i-1)N + 1, \ldots, i N$}
                \State Sample $\tilde{X}_t^{j} \sim q(x_t \mid y_{1:t}, X^j_{t-1}, \theta^i)$
                \State Compute $w(\tilde{X}_t^{j}) := \frac{g(y_t \mid \tilde{X}_t^{j}) f(\tilde{X}_t^{j} \mid X^j_{t-1}, y_{t-1}, \theta^i)}{q(x_t \mid y_{1:t}, X^j_{t-1}, \theta^i)}$
            \EndFor
            \State Compute $\hat{p}(y_t \mid y_{1:t-1}, \theta^i) := N^{-1} \sum_{i=1}^N
        w(\tilde{X}_t^i)$ 
            \State Compute $\hat{\mathbb{E}}[f(x_t) \mid y_{1:t},\theta^i] = \sum_{j=1}^N \frac{w(\tilde{X}_t^j)}{ \sum_k w(\tilde{X}_t^k) } f(\tilde{X}_t^j)$ 
            \For{$j=(i-1)N + 1, \ldots, i N$}
                \State Sample $I_t^j$ with $P(I_t^j = k) \propto w(\tilde{X}_t^{k})$, $k \in \{ (i-1)N + 1, \ldots, i N\}$
                \State Assign $X_t^j := \tilde{X}_1^{I_t^j}$
            \EndFor
        \EndFor
        \State Compute $\hat{p}(y_t \mid y_{1:t-1}) := N_\theta^{-1} \sum_{i} \hat{p}(y_t \mid y_{1:t-1}, \theta^i)$
    \EndIf
\EndFor
\EndProcedure
\end{algorithmic}
}
\end{algorithm}


\section{Appendix B: MCMC Information and Diagnostics} \label{sec:appendixB}

First, the prior this is used for posterior estimation has been specified in section \ref{sec:posterior_samp_comp}. 

Regarding the parameter proposal distribution, we use a random walk multivariate normal distribution on a transformed parameter space for $(\text{logit}(\phi), \mu, \log \sigma^2, \text{logit}([\rho + 1]/2)$. The proposal covariance matrix is set to $\mathbf{I}\frac{2.38^2}{4}$. The likelihood evaluations are approximated with a SISR filter cite using $100$ particles and the state transition as a state proposal distribution (i.e. a ``bootstrap filter"). For each proposed parameter, $7$ particle filters are run in parallel, and their likelihoods are averaged together before it's plugged in to the Hastings ratio. 100,000 iterations are run. 

Each parameter yields an $\hat{R}$ diagnostic of less than $1.001$. The trace plots and correlograms are shown below in Figures \ref{fig:mcmc_mu_diagnostics}, \ref{fig:mcmc_phi_diagnostics}, \ref{fig:mcmc_ss_diagnostics}, and \ref{fig:mcmc_rho_diagnostics}.

\begin{figure}
     \centering
     \begin{subfigure}[b]{0.45\textwidth}
         \centering
         \includegraphics[width=\textwidth]{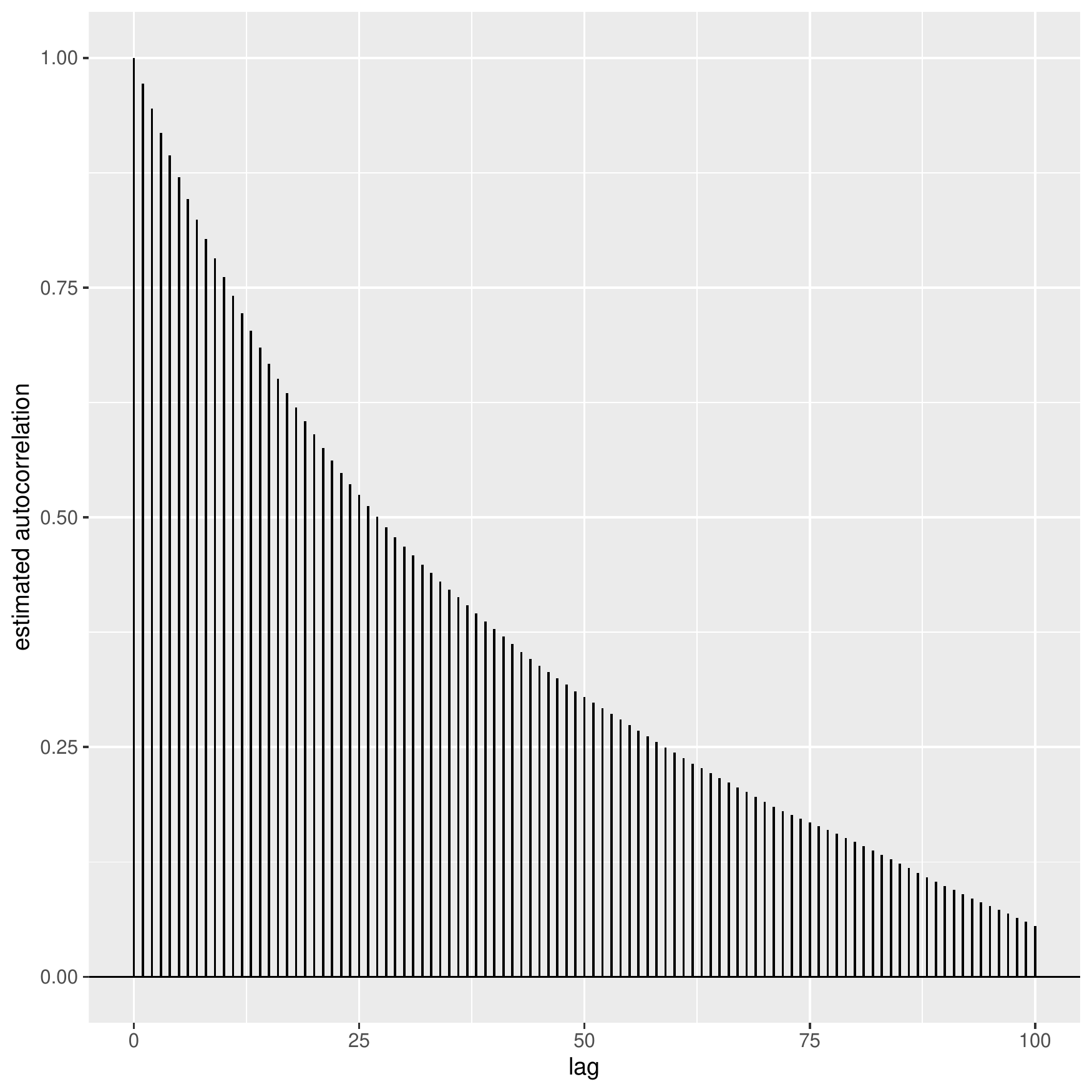}
         \caption{Autocorrelation plot for $\mu$ samples}
         \label{fig:mu_acf}
     \end{subfigure}
     \hfill
     \begin{subfigure}[b]{0.45\textwidth}
         \centering
         \includegraphics[width=\textwidth]{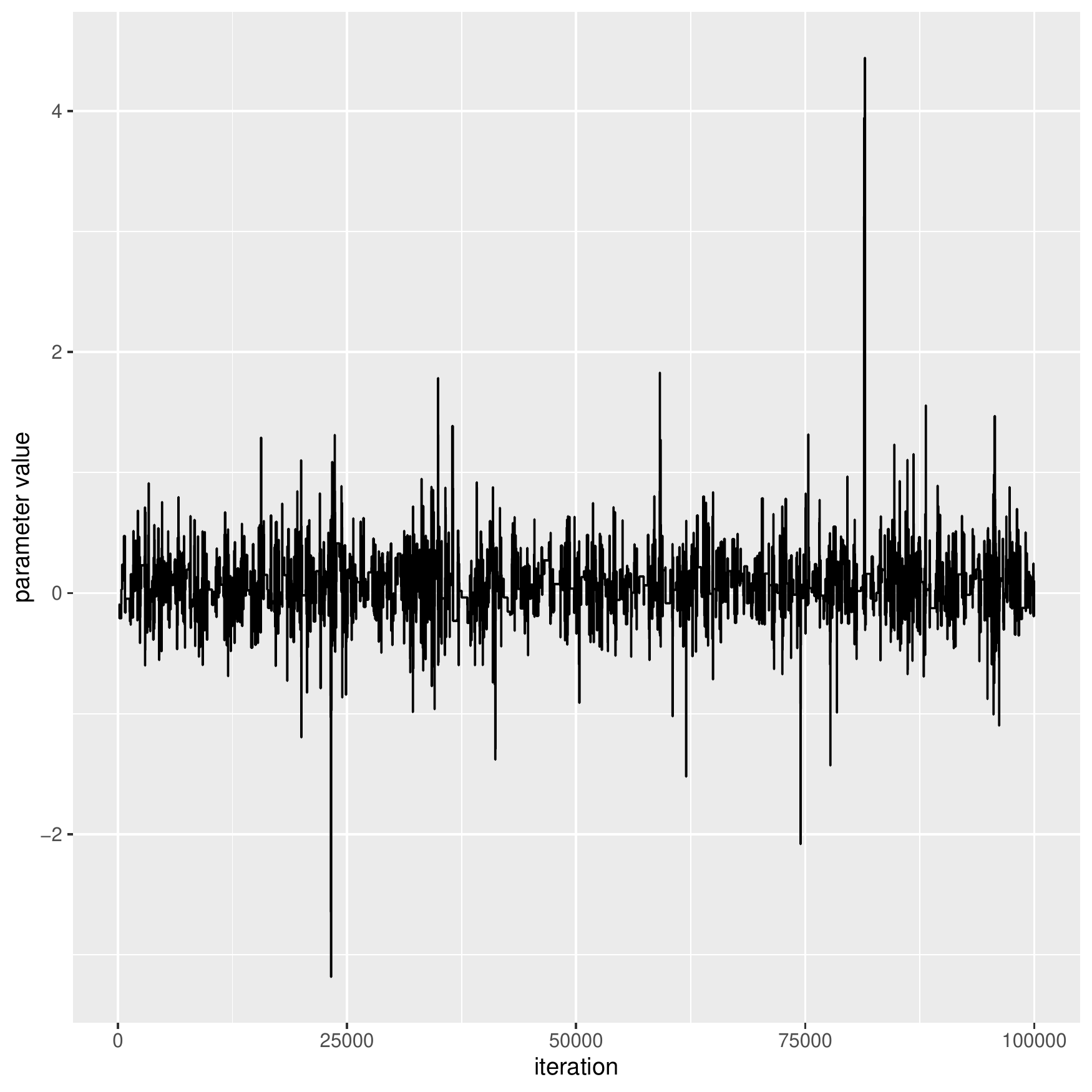}
         \caption{Trace plot for $\mu$ samples}
         \label{fig:mu_trace}
     \end{subfigure}
     \caption{MCMC diagnostic plots}
        \label{fig:mcmc_mu_diagnostics}
\end{figure}

\begin{figure}
     \centering
    \begin{subfigure}[b]{0.45\textwidth}
         \centering
         \includegraphics[width=\textwidth]{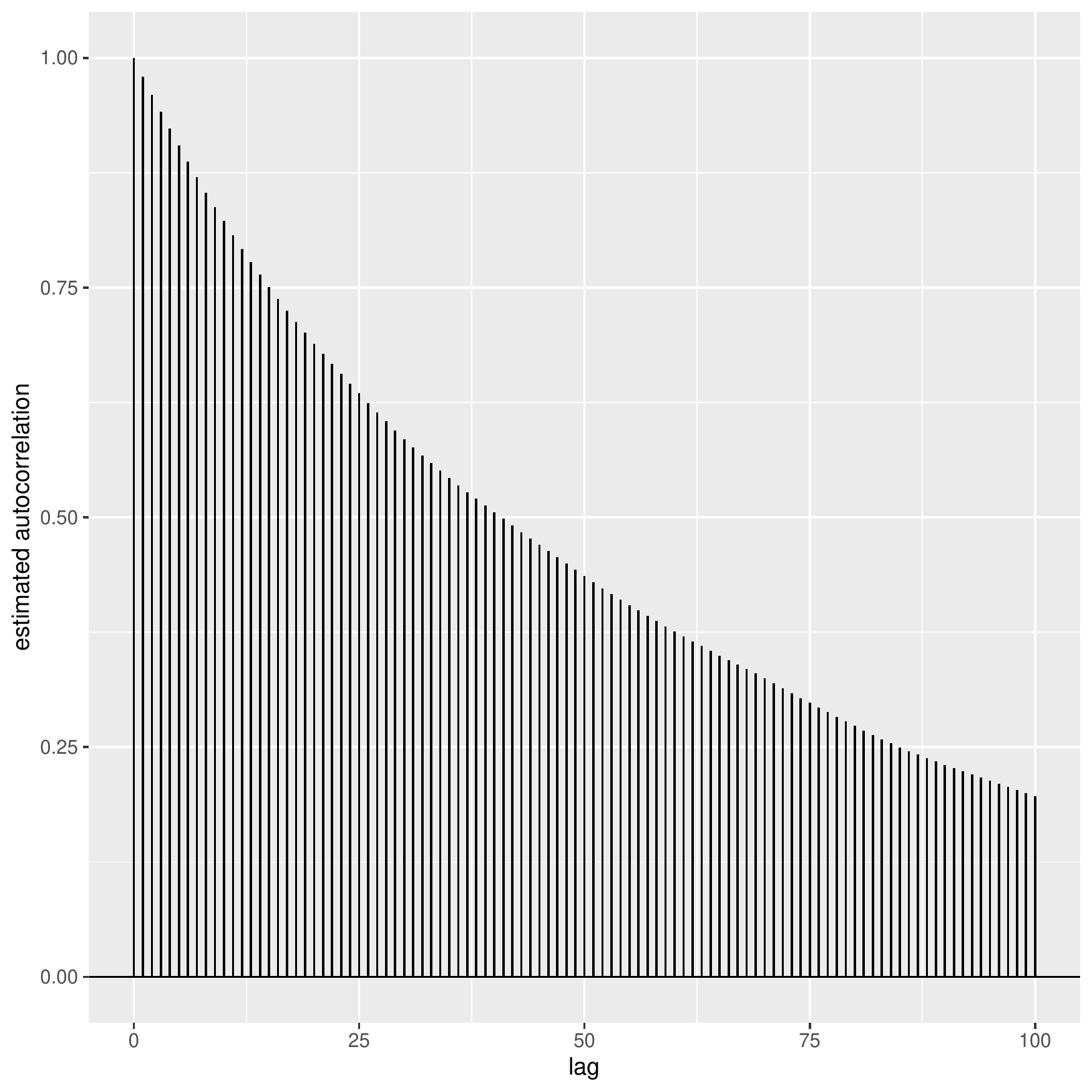}
         \caption{Autocorrelation plot for $\phi$ samples}
         \label{fig:phi_acf}
     \end{subfigure}
     \hfill
     \begin{subfigure}[b]{0.45\textwidth}
         \centering
         \includegraphics[width=\textwidth]{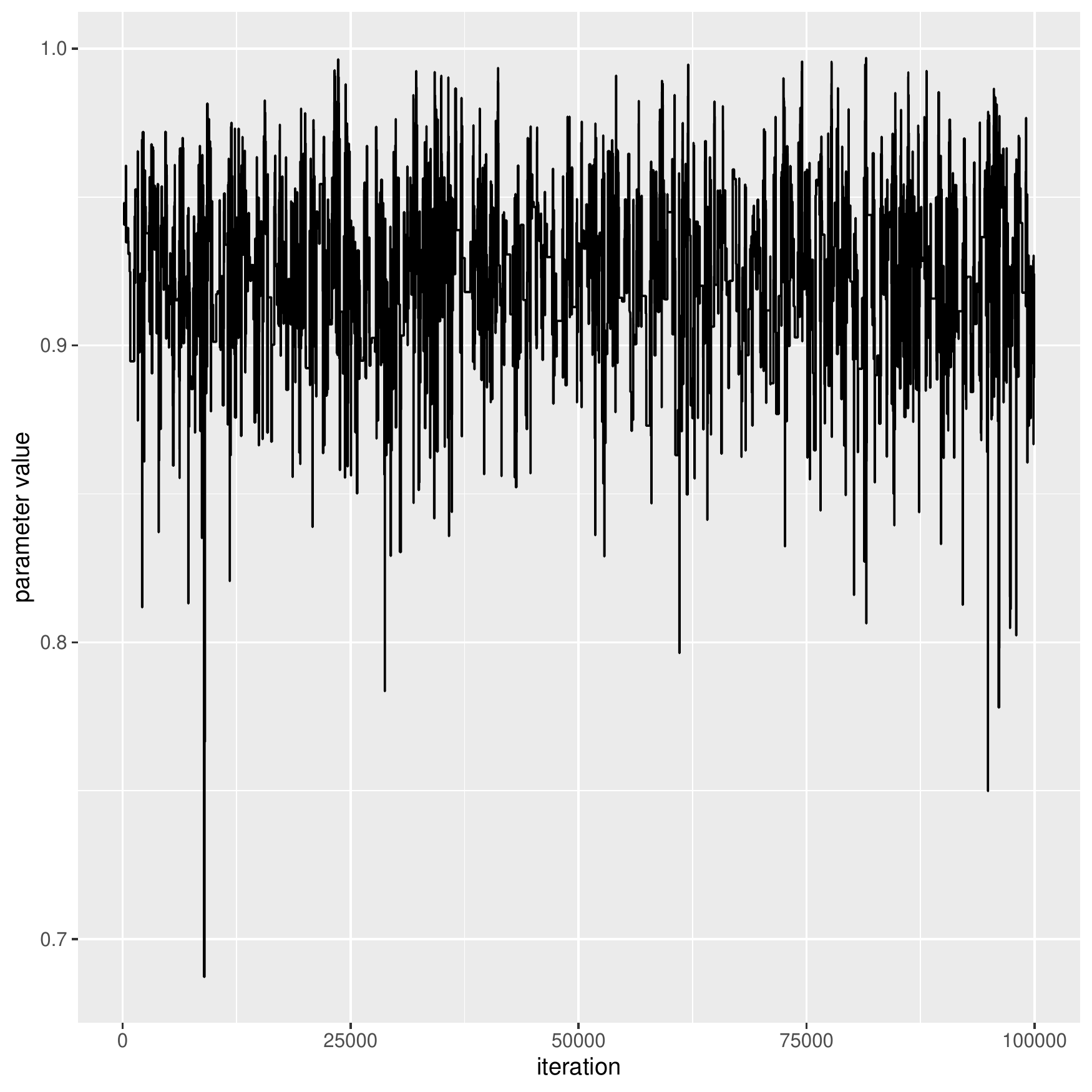}
         \caption{Trace plot for $\phi$ samples}
         \label{fig:phi_trace}
     \end{subfigure}
     \caption{MCMC diagnostic plots}
        \label{fig:mcmc_phi_diagnostics}
\end{figure}

\begin{figure}
     \centering
    \begin{subfigure}[b]{0.45\textwidth}
         \centering
         \includegraphics[width=\textwidth]{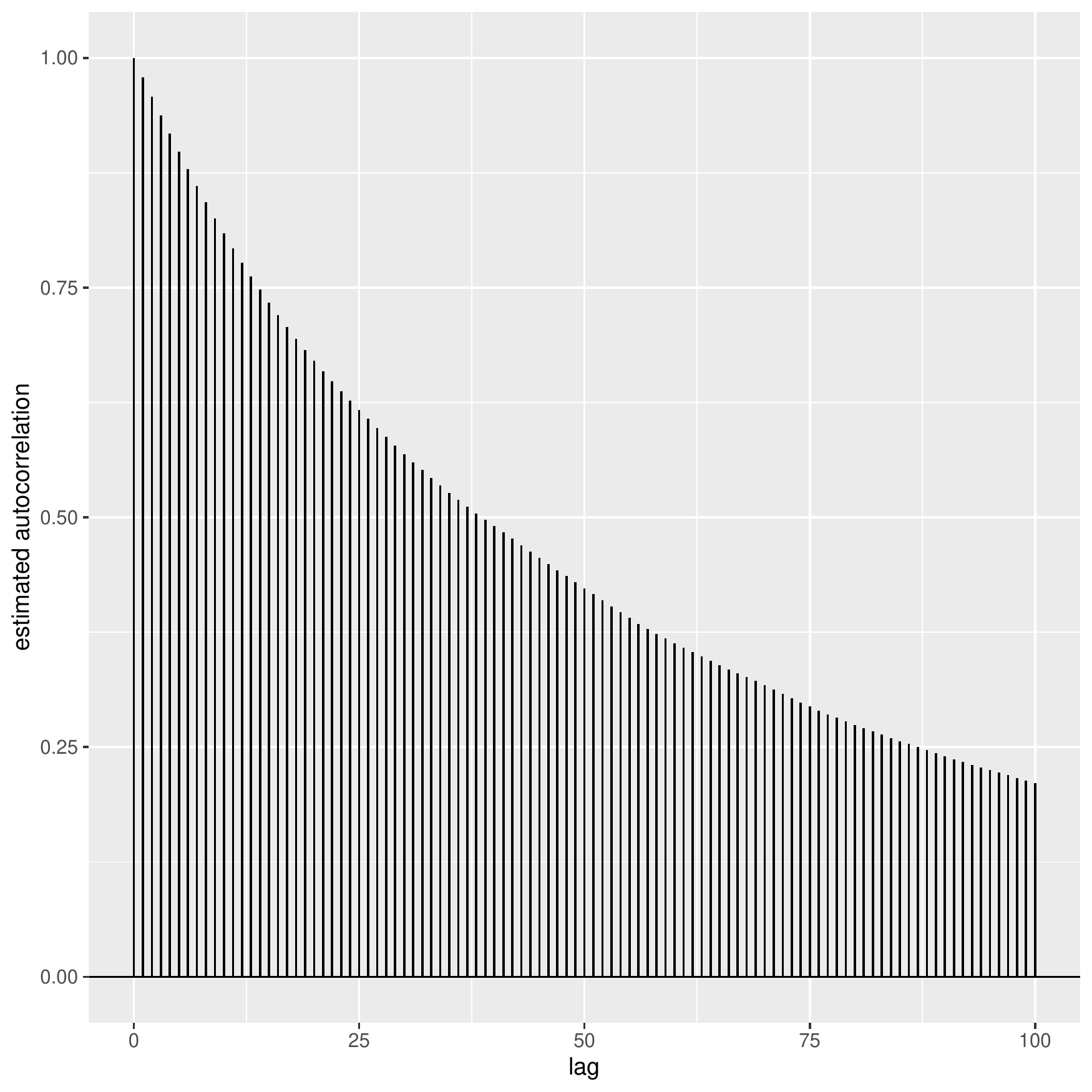}
         \caption{Autocorrelation plot for $\sigma^2$ samples}
         \label{fig:sigmaSq_acf}
     \end{subfigure}
     \hfill
     \begin{subfigure}[b]{0.45\textwidth}
         \centering
         \includegraphics[width=\textwidth]{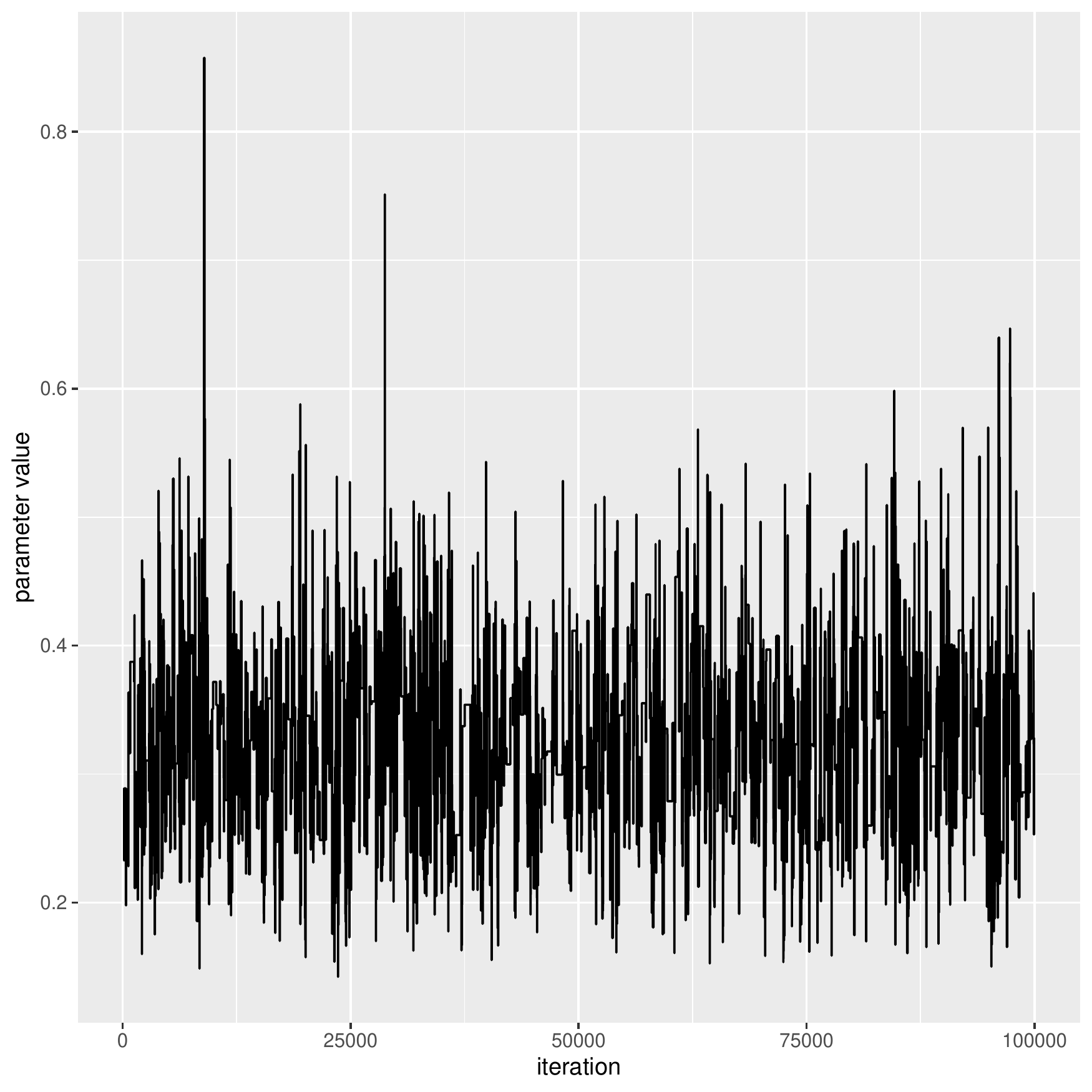}
         \caption{Trace plot for $\sigma^2$ samples}
         \label{fig:sigmaSq_trace}
     \end{subfigure}
     \caption{MCMC diagnostic plots}
        \label{fig:mcmc_ss_diagnostics}
\end{figure}

\begin{figure}
     \centering
    \begin{subfigure}[b]{0.45\textwidth}
         \centering
         \includegraphics[width=\textwidth]{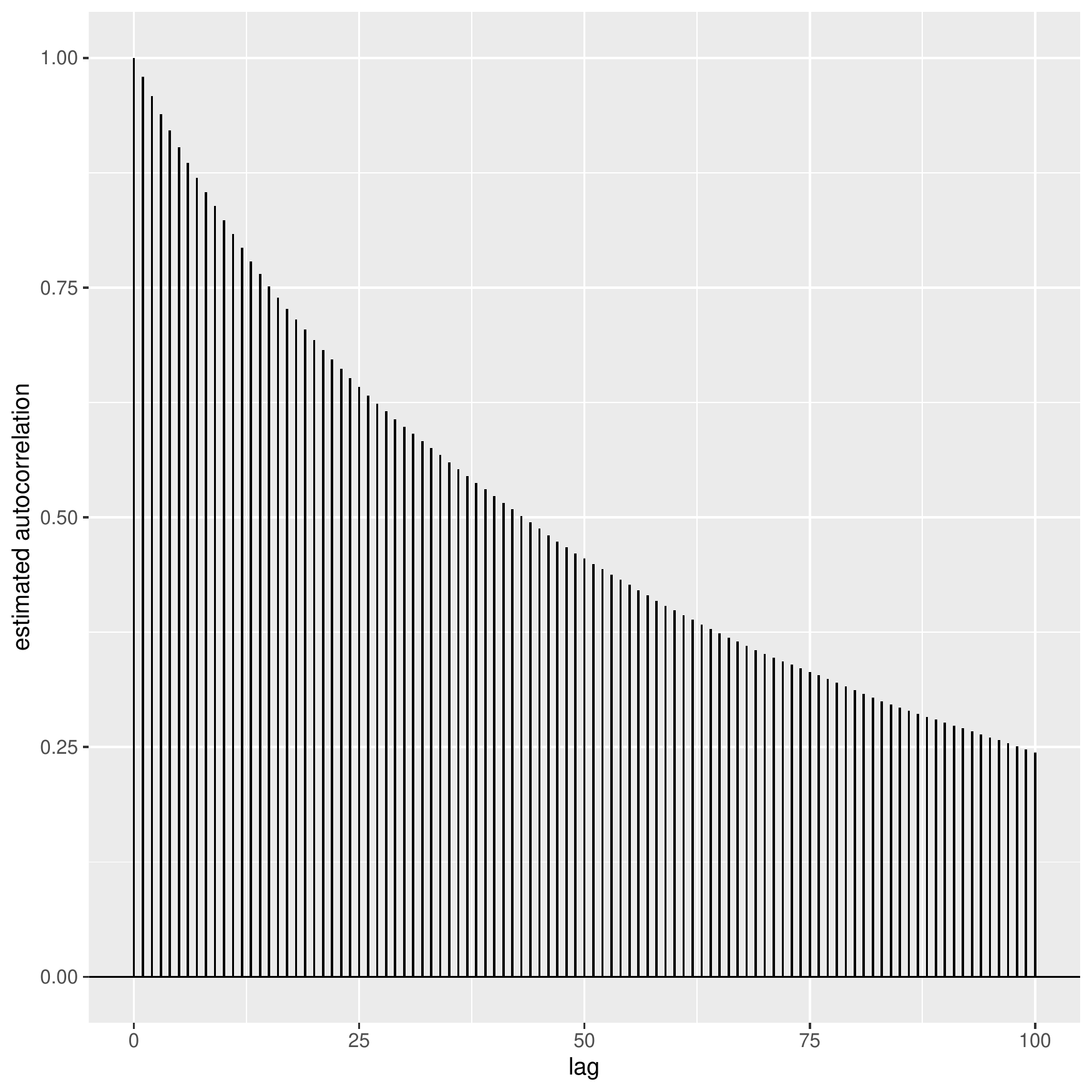}
         \caption{Autocorrelation plot for $\rho$ samples}
         \label{fig:rho_acf}
     \end{subfigure}
     \hfill
     \begin{subfigure}[b]{0.45\textwidth}
         \centering
         \includegraphics[width=\textwidth]{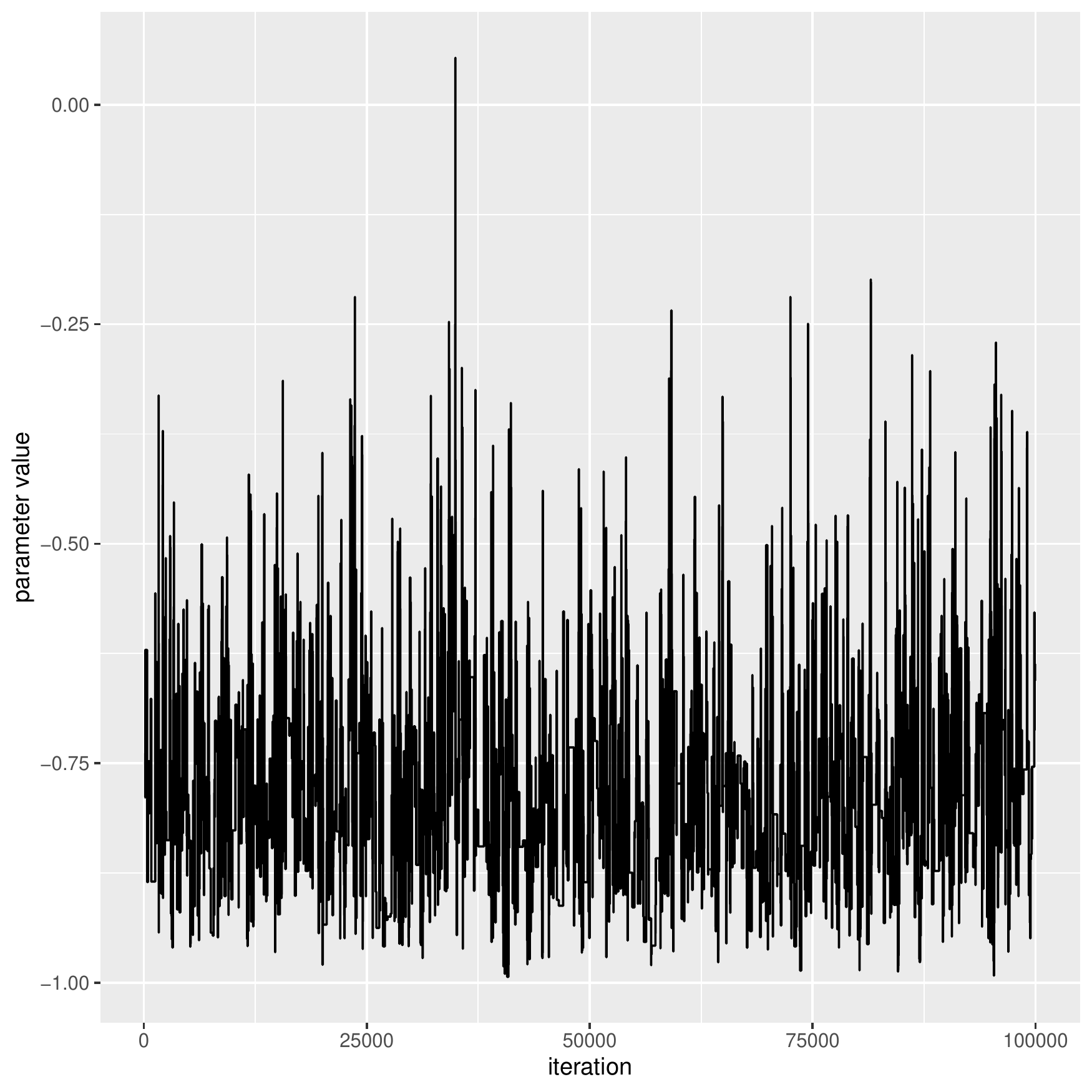}
         \caption{Trace plot for $\rho$ samples}
         \label{fig:rho_trace}
     \end{subfigure}
     \caption{MCMC diagnostic plots}
        \label{fig:mcmc_rho_diagnostics}
\end{figure}


\end{appendices}

\bibliographystyle{plain} 
\bibliography{mybib} 

\end{document}